\RequirePackage[l2tabu, orthodox]{nag}
\documentclass{article}

\usepackage[T1]{fontenc}

\usepackage[bitstream-charter]{mathdesign}
\usepackage{bbm}
\usepackage{amsmath}
\usepackage[scaled=0.92]{PTSans}
\usepackage{csquotes}

\usepackage[
  paper  = letterpaper,
  left   = 1.2in,
  right  = 1.2in,
  top    = 1.0in,
  bottom = 1.0in,
  ]{geometry}

\usepackage[dvipsnames]{xcolor}
\definecolor{shadecolor}{gray}{0.9}

\usepackage[final,expansion=alltext]{microtype}
\usepackage[parfill]{parskip}
\usepackage[english]{babel}

\usepackage{graphicx}
\usepackage[labelfont=bf]{caption}
\usepackage[format=hang]{subcaption}
\usepackage{sidecap}

\usepackage{booktabs}
\usepackage{tabularx}

\usepackage{authblk}

\usepackage[round, sort]{natbib}

\usepackage[colorlinks,linktoc=all]{hyperref}
\usepackage[all]{hypcap}
\hypersetup{citecolor=BrickRed}
\hypersetup{linkcolor=black}
\hypersetup{urlcolor=MidnightBlue}

\usepackage{tcolorbox}

\usepackage{listings}

\usepackage{amsthm}
\newtheorem{theorem}{Theorem}[section]
\newtheorem{lemma}[theorem]{Lemma}
\newtheorem{proposition}[theorem]{Proposition}

\newtheorem{definition}[theorem]{Definition}
\theoremstyle{definition}

\newtheorem{assumption}[theorem]{Assumption}

\usepackage{enumitem}

\usepackage{tikz} 

\usepackage{pifont}

\newcommand{\RR}{\mathbbmss{R}} 
\newcommand{\EE}[1]{\mathbbmss{E}\left[#1\right]} 
\newcommand{\PP}[1]{\text{Pr}\left[#1\right]} 
\newcommand{\VV}[1]{\mathbbmss{V}\left[#1\right]} 
\newcommand{\one}{\mathbbmss 1} 
\newcommand{\Normal}[2]{\mathcal{N}\left(#1,#2\right)} 
\newcommand{\pto}{\overset{p}{\rightarrow}} 
\newcommand{\dto}{\leadsto} 

\newcommand{\data}{\mathcal{D}}

\newcommand{\indep}{\perp \!\!\!\! \perp}
\newcommand{\empEE}[2]{\hat{\mathbbmss{E}}^{#1}\left[#2\right]} 
\newcommand{\Ltwo}{\mathbbmss{L}_2} 
\newcommand{\norm}[1]{\left\|#1\right\|_{\Ltwo}} 

\newcommand{\thetatarget}{\theta^\star}

\newcommand{\pitarget}{\pi^\star}
\newcommand{\mutarget}{\mu^\star}
\newcommand{\gammatarget}{\gamma^\star}

\newcommand{\FOCAL}{\texttt{FOCaL}}

\title{\textbf{A Doubly Robust Machine Learning Approach for Disentangling Treatment Effect Heterogeneity with Functional Outcomes}}

\author[1,2*]{Filippo Salmaso}
\author[1,3*]{Lorenzo Testa}
\author[1,4$\#$]{Francesca Chiaromonte}

\affil[1]{L'EMbeDS, Sant'Anna School of Advanced Studies, Pisa, Italy}
\affil[2]{School of Pharmaceutical Science, University of Geneva, Genève, Switzerland}
\affil[3]{Department of Statistics \& Data Science, Carnegie Mellon University, Pittsburgh PA, US}
\affil[4]{Department of Statistics, Penn State University, University Park PA, US}

\affil[*]{\small{These authors contributed equally}}
\affil[$\#$]{\small{Corresponding author}}

\date{\vspace{-0.1em}
\quad \texttt{filippo.salmaso@unige.ch} 
\quad \texttt{lorenzo@stat.cmu.edu} 
\quad \texttt{fxc11@psu.edu}
\\ 
\vspace{1em}\today}

\begin{document}

\maketitle

\begin{abstract}
Causal inference is paramount for understanding the effects of interventions, yet extracting personalized insights from increasingly complex data remains a significant challenge for modern machine learning. This is the case, in particular, when considering functional outcomes observed over a continuous domain (e.g.,~time, or space). Estimation of heterogeneous treatment effects, known as CATE, has emerged as a crucial tool for personalized decision-making, but existing meta-learning frameworks are largely limited to scalar outcomes, failing to provide satisfying results in scientific applications that leverage the rich, continuous information encoded in functional data. Here, we introduce \FOCAL{} (Functional Outcome Causal Learning), a novel, doubly robust meta-learner specifically engineered to estimate a \textit{functional}  heterogeneous treatment effect (F-CATE). \FOCAL{} integrates advanced functional regression techniques for both outcome modeling and functional pseudo-outcome reconstruction, thereby enabling the direct and robust estimation of F-CATE. We provide a rigorous theoretical derivation of \FOCAL{}, demonstrate its performance and robustness compared to existing non-robust functional methods through comprehensive simulation studies, and illustrate its practical utility on diverse real-world functional datasets. \FOCAL{} advances the capabilities of machine intelligence to infer nuanced, individualized causal effects from complex data, paving the way for more precise and trustworthy AI systems in personalized medicine, adaptive policy design, and fundamental scientific discovery.
\end{abstract}

\section{Introduction}

Causal inference has become a cornerstone of modern statistics and machine learning, providing a formal language and a rigorous toolkit for moving beyond associative relationships to answer substantial questions about the effects of interventions in a variety of fields, including medicine \citep{prosperi2020causal}, genomics \citep{du2025causal}, economics \citep{varian2016causal}, and the social sciences \citep{imbens2024causal}. Early developments in this field \citep{rubin1974estimating, rubin2005causal} established the so-called \textit{counterfactual} framework for estimating the \textit{Average Treatment Effect} (ATE), which measures the mean effect of a treatment across an entire population. 

Although the ATE is a valuable summary 
estimand, it often masks significant underlying variation. In many real-world scenarios, the effect of a treatment is not uniform but varies across individuals with different characteristics. For example, a new drug may be highly effective for one patient subgroup but have little to no effect on another. This phenomenon, known as \textit{treatment effect heterogeneity}, has become a critical frontier in causal inference. Understanding ``what works for whom'' is essential, e.g., for personalized medicine and targeted policy implementation \citep{xu2024optimal}. The central statistical target for quantifying this heterogeneity is the \textit{Conditional Average Treatment Effect} (CATE), which measures the average treatment effect for subpopulations defined by a specific set of covariates. Estimating the CATE allows one to map how the treatment effect changes across the covariate space.

The challenge of estimating the CATE has been tackled through the development of effective and flexible machine learning methods known as \textit{meta-learners} \citep{molak2023causal, caron2022estimating, wager2024causal}. These algorithms leverage the predictive power of modern machine learning, typically by reframing the causal question as a sequence of supervised learning problems. Using non-parametric models such as random forests, gradient boosting, or neural networks one can capture complex relationships in the data without imposing strong parametric assumptions on the CATE structure. Thus, machine learning techniques -- including black-boxes -- can facilitate inference on nuanced causal effects.

Concurrently, the nature of the data being collected in many scientific fields has evolved in complexity, presenting compelling new frontiers for machine intelligence. In particular, advances in technology have led to an explosion of \textit{functional data}, where each observation is not a single number or a vector, but an entire function recorded over a continuum like time or space \citep{kokoszka2017introduction, ramsay2005}. Examples include longitudinal studies where medical or wearable devices are used to record patients' biometrics over time \citep{boschi2024new, jeong2024identifying}, epidemiological studies where infectious diseases are tracked across spatiotemporal domains \citep{boschi2021functional, boschi2023contrasting}, neuroscience studies where activity is monitored over time and across the volume of the brain \citep{qi2018function, boschi2024fasten}, and economic and financial studies where metrics of interest are measured over time \citep{esposito2022venture, caldeira2017forecasting, cremona2023functional}. This type of data poses unique challenges and offers unique opportunities for robust and adaptive machine learning approaches. \textit{Functional Data Analysis} (FDA) provides a set of tools to analyze such data, properly accounting for the inherent smoothness and structure of curves and surfaces, and avoiding the pitfalls of treating them as simple high-dimensional vectors -- a crucial step for developing intelligent systems that can truly understand and reason about complex, continuous processes.

Our work sits at the intersection of causal inference and FDA; while both these areas of statistics are quite advanced and rapidly evolving, machine learning methods to address 
causal questions within the context of functional data remain underdeveloped. A significant gap exists in the literature for methods to estimate heterogeneous treatment effects with functional outcomes, and this gap hinders the development of truly personalized and adaptive AI solutions. Traditional meta-learner methods for estimating heterogeneous treatment effects cannot be applied directly, and current functional data techniques are 
either designed for settings without access to additional covariates \citep{cremona2018iwtomics}, or model outcomes with non-robust tools such as function-on-scalar regression \citep{ecker2024causal, ieva2025enhancing}
which may perform poorly if the model is misspecified, limiting their utility in real-world applications. Relevant exceptions, limited to the estimation of average treatment effects, include \citet{liu2024double, raykov2025kernel, testa2025doubly}, who provide estimators that are robust to some forms of model misspecification. This work aims to bridge the gap by introducing a novel meta-learner that leverages advanced machine learning techniques for F-CATE estimation; that is, estimation of the CATE in settings with functional outcome and scalar predictors.

Our meta-learner, which we call \FOCAL{} (Functional Outcome Causal Learning), enjoys double robustness and significantly advances the capabilities of machine intelligence in these complex settings. We provide a careful derivation of \FOCAL{}, describe its statistical properties, and demonstrate its practical advantages through a comprehensive simulation study. We also showcase its practical utility through the analysis of two datasets: the SHARE dataset \citep{mannheim2005survey}, where we study the heterogeneous effects of chronic conditions on the time progression of quality of life indicators, and a dataset on the COVID-19 epidemic in Italy \citep{boschi2023contrasting}, where we investigate the causal implications of distributed primary health care on COVID-19 mortality patterns.

This work provides a rigorous foundation for applying modern causal learning techniques to the rich data structures increasingly encountered in scientific research, paving the way for more sophisticated and trustworthy AI systems capable of deep causal understanding and precision decision-making.

\section{Results}

\subsection{Overview of \FOCAL{}}
Building upon the principles of doubly robust meta-learning \citep{kennedy2023towards}, \FOCAL{} estimates (and quantifies the uncertainty behind) the causal effect of a binary treatment on a functional outcome, controlling for (i.e.,~conditioning on) measured categorical and scalar confounders, and thus capturing heterogeneity in causal effects. \FOCAL{} requires some standard causal inference assumptions, including conditional ignorability and positivity \citep{tsiatis2006semiparametric}, 
and operates on a collection of $n$ i.i.d.~samples $\data_i=(X_i,A_i,\mathcal{Y}_i)$, $i=1,\dots,n$, where $X_i$ is a set of pretreatment scalar and/or categorical covariates, $A_i$ is a binary variable indicating whether observation $i$ has been exposed to a treatment ($A_i=1$) or not ($A_i=0$), and $\mathcal{Y}_i$ is a functional outcome measured over a continuous domain. Our approach can be applied regardless of the dimension of this domain, but for simplicity we will consider a one-dimensional domain such as time -- so we can think of $\mathcal{Y}_i$ as a smooth curve satisfying standard FDA assumptions \citep{kokoszka2017introduction}.
Based on the available samples, \FOCAL{} returns an estimate of $\thetatarget(x)$, the functional conditional average treatment effect (F-CATE). For each specification $x$ of the covariates, the F-CATE is itself a function unfolding on the same domain as the outcomes $\mathcal{Y}_i$'s;
its dependence on $x$ captures heterogeneity in the response to treatment between individuals with different characteristics (formal definitions, along with more detail on assumptions and theoretical derivations, are provided in the Methods).

\begin{SCfigure}[]
    \centering
    \includegraphics[width=0.67\linewidth]{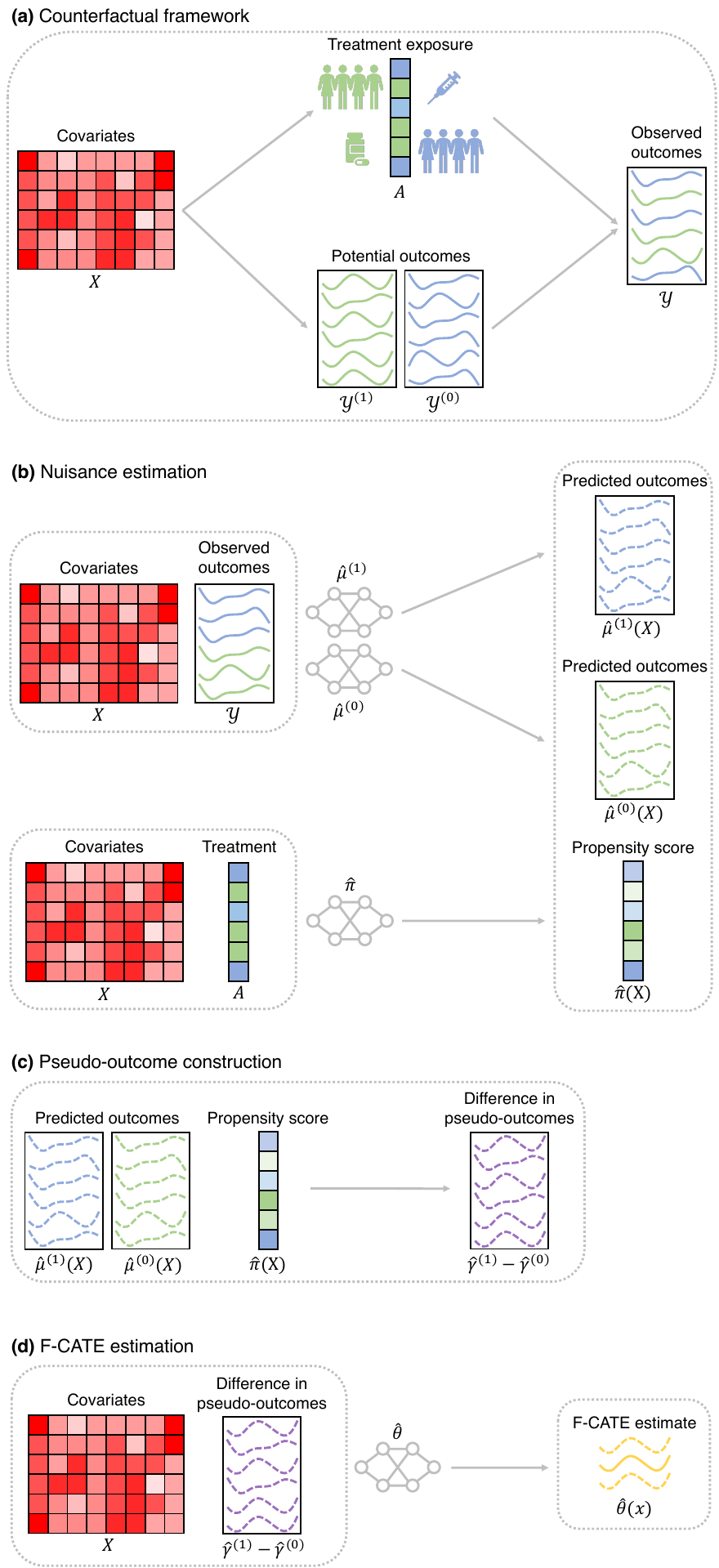}
    \caption{An overview of \FOCAL{}.
    \textbf{(a)} The counterfactual framework. For every subject defined by the covariates in $X$, two potential functional outcomes exist: $\mathcal{Y}^{(1)}$ (outcome under treatment) and $\mathcal{Y}^{(0)}$ (outcome under control). The treatment assignment $A$ acts as a selector, determining which of these potential trajectories is realized as the observed outcome $\mathcal{Y}$, leaving the other unobserved (counterfactual). \textbf{(b)} Nuisance function estimation (\FOCAL{} stage 1). Flexible machine learning models are employed to estimate the conditional outcome regressions $\hat{\mu}^{(a)}(X)$ for 
    the two treatment arms $a = 0,1$, and the propensity score $\hat{\pi}(X)$. \textbf{(c)} Pseudo-outcome construction (\FOCAL{} stage 2). The nuisance estimates are combined to reconstruct functional pseudo-outcomes using a doubly robust formulation, isolating the difference $\hat{\gamma}^{(1)} - \hat{\gamma}^{(0)}$ adjusted for confounding. \textbf{(d)} F-CATE estimation and inference (\FOCAL{} stage 3). The difference in pseudo-outcomes is regressed on the covariates to yield the final estimator $\hat{\theta}(x)$, which quantifies the heterogeneous functional treatment effect across the covariate space. Dashed lines represent simultaneous confidence bands computed around the estimated F-CATE
    to capture uncertainty.}
    \label{fig:overview}
\end{SCfigure}

\FOCAL{}'s core pipeline consists of three interconnected stages (see Figure~\ref{fig:overview}). In the first, \FOCAL{} leverages flexible machine learning algorithms to estimate from the data some crucial, so-called \textit{nuisance} functions \citep{tsiatis2006semiparametric}: the \textit{regression functions} ${\mutarget}^{(a)}(x)$, $a\in\{0,1\}$, which represent the expected functional outcome given the covariates and the treatment assignment, and the \textit{propensity score} $\pitarget(x)$, which represents the probability of receiving the treatment given the covariates. The estimates $\hat\mu^{(a)}(x)$, $a\in\{0,1\}$ are obtained with advanced functional regression techniques \citep{boschi2024fasten, boschi2024fungcn, yao2021deep}, where the covariates map into an entire functional response over the continuous domain under consideration, allowing \FOCAL{} to learn the expected shape and trajectory of the outcome. The estimate $\hat\pi(x)$ is typically handled as a standard binary classification problem; a number of powerful classifiers can be employed to yield an estimate of $\PP{A=1\mid X=x}$. The use of flexible, data-driven machine learning models in these initial estimations is fundamental, as it allows \FOCAL{} to accurately capture complex, non-linear relationships without imposing rigid parametric assumptions, and integrating information coming from different data modalities (scalar and categorical covariates). 

In the second stage, \FOCAL{} reconstructs two functional pseudo-outcomes, ${\gammatarget}^{(a)}(\data_i)$, $a\in\{0,1\}$, for each individual $i=1,\ldots,n$. These pseudo-outcomes are specifically defined to isolate the causal effect of the treatment on the functional outcome, adjusting for potential confounding. They are estimated as:
\begin{equation}
    \hat\gamma^{(1)}(\data_i) = \hat\mu^{(1)}(X_i) + \frac{A_i}{\hat\pi(X_i)}\left(\mathcal{Y}_i - \hat\mu^{(1)}(X_i) \right)\,,\quad \hat\gamma^{(0)}(\data_i) = \hat\mu^{(0)}(X_i) + \frac{1-A_i}{1-\hat\pi(X_i)}\left(\mathcal{Y}_i - \hat\mu^{(0)}(X_i) \right)\,.
\end{equation}
This formulation endows \FOCAL{} with its crucial \textit{doubly robust} property: the estimators for the functional pseudo-outcomes remain asymptotically consistent even if either the functional outcome modeling that produces $\hat\mu^{(a)}(x)$, $a\in\{0,1\}$, or the propensity score modeling that produces $\hat\pi(x)$ are misspecified, provided one of the two is correctly specified. This inherent robustness is a significant advantage for machine intelligence systems deployed in real-world scenarios, where achieving perfect model specification is often impractical. In particular, in the context of our proposal, double robustness enhances the reliability of
causal inferences drawn from functional data.

In the third stage, \FOCAL{} estimates the F-CATE $\thetatarget(x)$. This is accomplished by performing a final regression of the difference in the computed functional pseudo-outcomes $\hat\gamma^{(1)}(\data) - \hat\gamma^{(0)}(\data)$ on the covariates. This step utilizes again advanced functional regression techniques to model how the entire functional treatment effect changes across the covariate space. The output is a function $\hat\theta(x)$ that, for any given covariate profile $x$, describes the expected difference in the functional outcome due to treatment and quantifies its uncertainty by providing a \textit{simultaneous confidence band} which contains the true $\thetatarget(x)$ with at least a user-specified probability. Through this process, \FOCAL{} moves beyond scalar averages, providing granular, reliable, and interpretable causal insights that pinpoint not only who benefits from an intervention, but precisely how their functional response (e.g., the evolution of a quality of life indicator, the progression of mortality during an epidemic
) is expected to change. \FOCAL{} thus empowers machine intelligence to reason effectively about interventions in complex, dynamic systems.

\subsection{Simulation results provide strong evidence of \FOCAL{}'s effectiveness}

\begin{figure}[ht!]
    \centering
    \includegraphics[width=\linewidth]{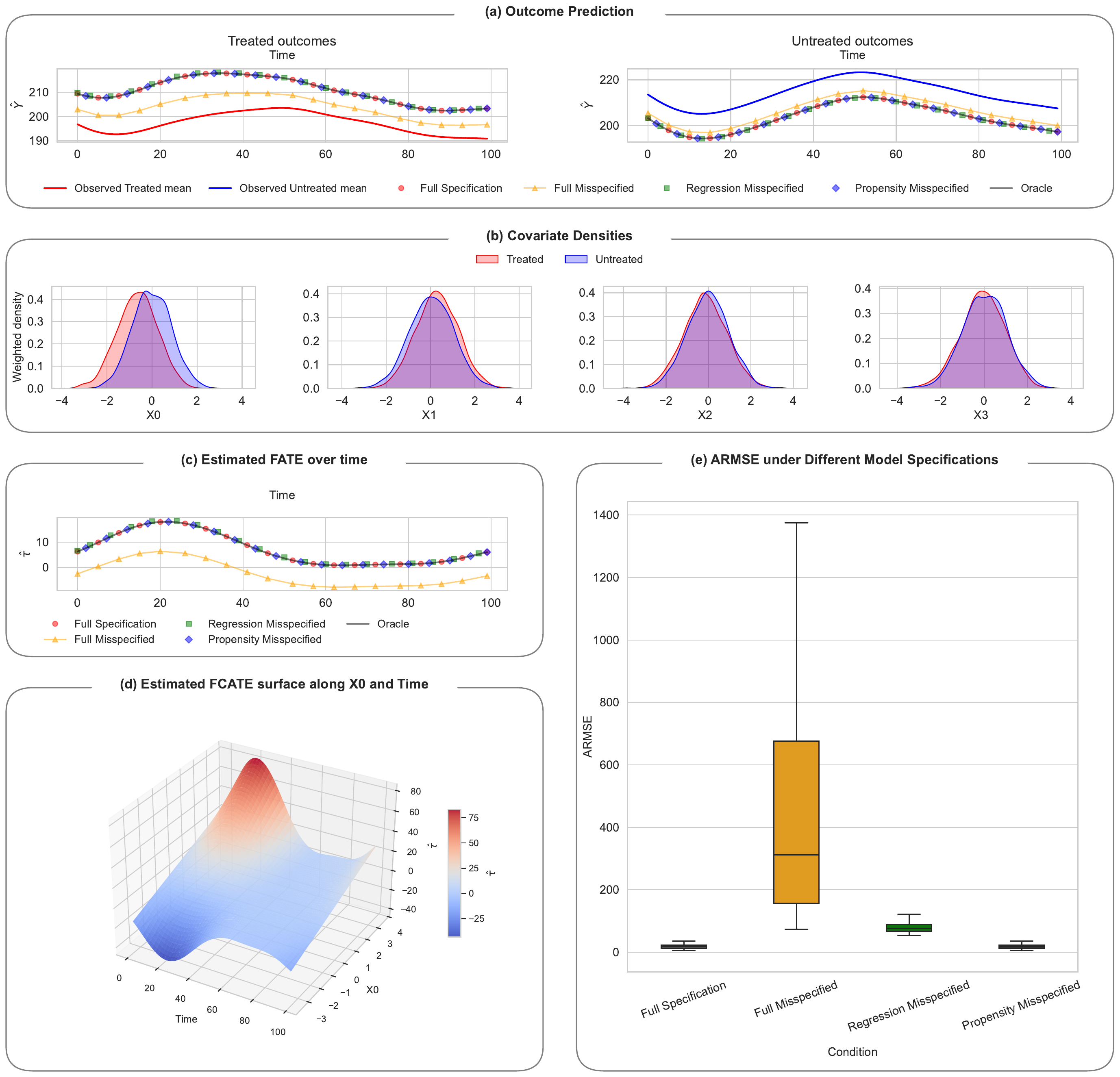}
    \caption{Performance and robustness of \FOCAL{} in simulated settings.
    \textbf{(a)} Comparison of functional outcome trajectories for treated (left) and untreated (right) populations. Solid lines represent the naive observed means, which exhibit significant bias due to confounding (departures from the Oracle, which is shown in solid gray).
    Lines with symbols (circle, triangle, squares, diamonds) represent the model-recovered trajectories under different specifications; note that \FOCAL{} accurately reconstructs the 
    true curves even when one nuisance component is misspecified, correcting the bias inherent in the raw data.
    \textbf{(b)} Density plots for original covariates ($X_j$, $j=0,1,2,3$) illustrating the distributional imbalance between treated (red) and control (blue) groups,
    which motivates the adjustment (the distributions differ especially for $X_0$).
    \textbf{(c)} Evaluation of FATE estimation accuracy. The line plot displays the estimated Functional Average Treatment Effect (FATE) over time; estimates remain stable and close to the truth (Oracle, in solid grey) 
    unless both models are misspecified.
    \textbf{(d)} Example of F-CATE target. The surface plot visualizes the estimated Functional Conditional Average Treatment Effect (F-CATE) as a function of time and covariate
    $X_0$, demonstrating \FOCAL{}’s ability to capture complex, heterogeneous effect surfaces
    (the true surface is shown for reference in Supplementary Figure~\ref{fig:true_surface}).
    \textbf{(e)} Evaluation of F-CATE estimation accuracy. The boxplots show 
    Aggregated Root Mean Squared Error (ARMSE) values across 1000 Monte Carlo replications, confirming that 
    sizeable error inflation occurs only under
    double misspecification.
    }
    \label{fig:sim}
\end{figure}

To rigorously evaluate \FOCAL{}’s ability to recover
heterogeneous functional treatment effects, we conduct a comprehensive simulation study designed to mimic challenging real-world conditions. We generate datasets of $n=5000$ observations with time-varying functional outcomes defined as linear functions of coefficients with a Matérn covariance structure that induces realistic smoothness and correlation. Crucially, the data generating process features non-linear confounding between covariates and treatment assignment, creating a 
setting where simple linear adjustments would fail. We test the estimator across four distinct scenarios -- ranging from 
fully accurate model specification to simultaneous misspecification of both the propensity score and outcome regression components. 
These scenarios are created 
introducing non-linear transformations of the covariates into the nuisance models.

Our analysis first highlights the significant risks of relying on standard observational metrics in such complex settings. As visualized in Figure~\ref{fig:sim}\textbf{a}, naive comparisons of the observed mean trajectories between treated and untreated populations exhibit substantial bias, failing to account for the underlying confounding structure (due to covariate imbalance, see Figure~\ref{fig:sim}\textbf{b}). In contrast, \FOCAL{} effectively corrects for these biases. Even when nuisance models are imperfect, \FOCAL{} reconstructs the counterfactual trajectories with high fidelity, realigning the estimated effects, obtained by averaging the difference in pseudo-outcomes computed through \FOCAL{}'s second stage, with the ground truth and demonstrating the necessity of robust causal adjustments over simple associative comparisons (Figure~\ref{fig:sim}\textbf{c}).

The primary strength of \FOCAL{} -- its double robustness -- is confirmed by the stability of its estimation accuracy under stress. We quantify performance using the Aggregated Root Mean Squared Error (ARMSE, see Methods for a definition) between the estimated and true Functional Conditional Average Treatment Effects (F-CATE). \FOCAL{} exhibits negligible performance deterioration when either the propensity score or the outcome regression is misspecified, effectively matching the accuracy achieved under full model specification. As shown in the ARMSE boxplots (Figure~\ref{fig:sim}\textbf{e}), significant error inflation and variance degradation emerge only in the ``worst-case'' scenario where both nuisance models are simultaneously misspecified.

Furthermore, \FOCAL{} successfully recovers the granular structure of treatment effect heterogeneity. The estimated F-CATE surfaces (Figure~\ref{fig:sim}\textbf{d}) accurately map how the treatment effect varies continuously across both the functional domain and the covariate space, capturing the non-linear interactions driven by the data generating process. These findings empirically corroborate the theoretical guarantees of our framework, suggesting that \FOCAL{} can provide reliable, nuanced causal insights in scientific applications where the true data-generating mechanisms are often complex and rarely known with certainty.

\subsection{\FOCAL{} reveals the heterogeneous effects of chronic conditions on quality of life}

\begin{figure}[ht!]
    \centering
    \includegraphics[width=\linewidth]{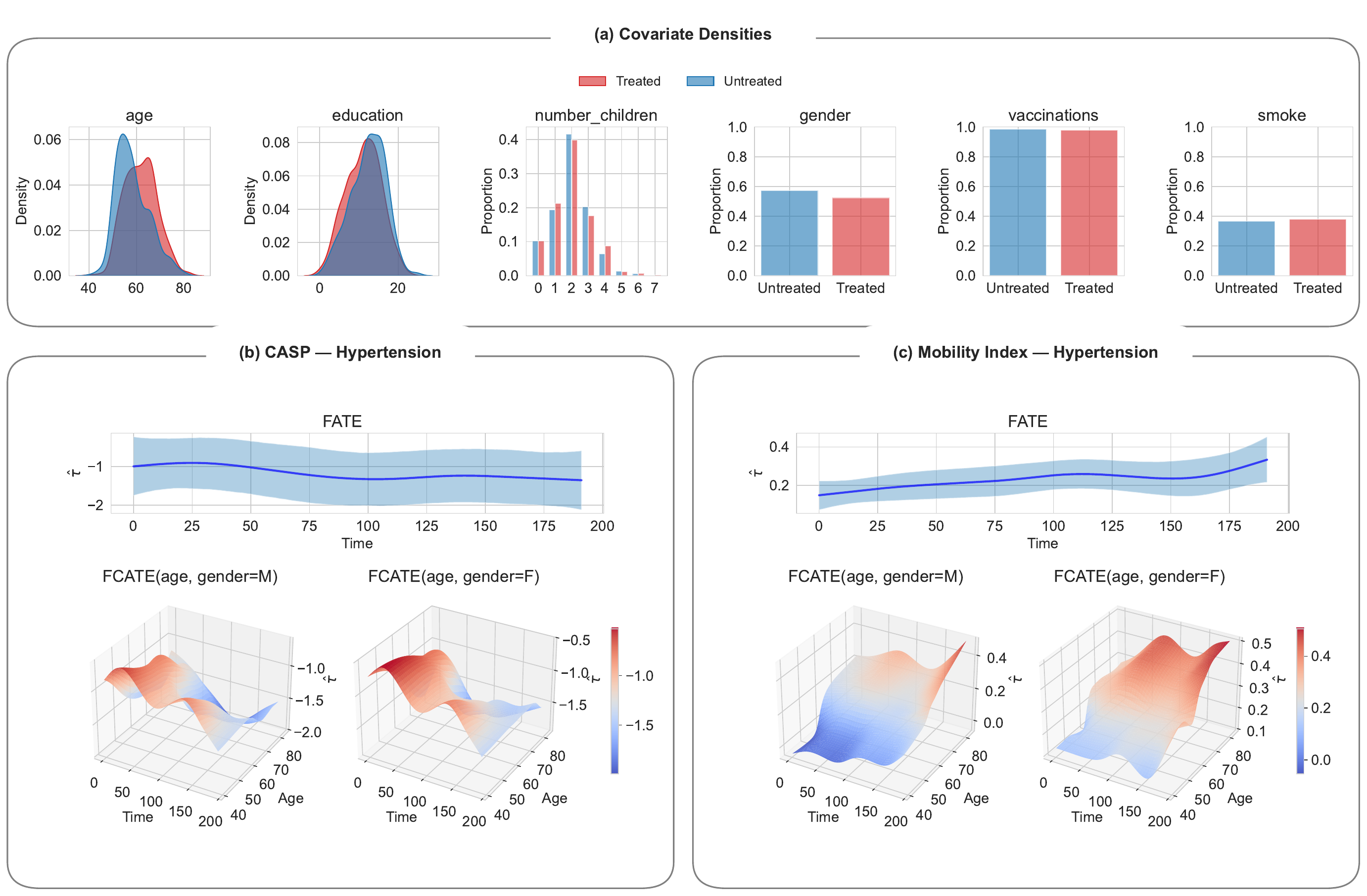}
    \caption{Heterogeneous impact of hypertension on quality of life trajectories (SHARE study). \textbf{(a)} Distributions of key baseline covariates for treated (hypertension) and control populations. While most socio-demographic factors show comparable spread, age exhibits a 
    sizeable imbalance between groups;
    this motivates the propensity adjustment performed by \FOCAL{}. \textbf{(b)} Results for CASP. The estimated Functional Average Treatment Effect (FATE) and the associated simultaneous 95\% confidence band over the 192-month observation window indicate a statistically significant decline in well-being: hypertension exerts a
    negative and progressively larger effect on the CASP quality of life score.
    The estimated F-CATE surfaces characterizing the heterogeneity of the causal effect conditioned on age and gender (all other binary and scalar covariates are set to their modal and mean values, respectively) show that the detrimental effect of hypertension intensifies with age 
    across the entire time domain. Gender plays an intercept-like role, amplifying the absolute magnitude of the functional effect for males without 
    markedly altering the shape of the 
    surface. \textbf{(c)} Results for Mobility Index. The estimated Functional Average Treatment Effect (FATE) and the associated simultenous 95\% confidence band show a positive, increasing, and statistically significant effect on the Mobility Index (higher values indicate greater impairment). 
    The estimated F-CATE surfaces reveal a distinct temporal profile for older individuals -- a moderate mid-period exacerbation
    followed by a sharp, late-stage increase during the final stages of the time domain. This profile is more pronounced and occurs earlier for women than for men.
    Note: 
    the plots for CASP and the Mobility Index
    have response-specific vertical scales
    (the two outcomes are measured on different scales; see Supplementary Figure \ref{fig:outhyp}).
    }
    \label{fig:share}
\end{figure}

To demonstrate the use of \FOCAL{}{}, we apply it to data from SHARE ({\em Survey of Health, Aging and Retirement in Europe}), a large longitudinal study comprising thousands of subjects followed through eight data collection surveys (called ``waves'') conducted between 2004 and 2020 \citep{mannheim2005survey, borsch2013data, bergmann2017survey, borsch2020survey}.

Specifically, we employ \FOCAL{} to investigate the heterogeneous causal effects of chronic conditions on longitudinal indicators of quality of life in different subpopulations. To ensure reliable curve estimation, we focus on $1518$ subjects who participated in at least seven (out of eight) waves (see Methods for full detail on data preprocessing). For each subject, we consider two functional indicators of quality of life -- a mobility index and the Quality of Life Scale (CASP) \citep{hyde2003measure} measured over 192 months; binary status relative to the chronic disease hypertension; and a variety of socio-demographic and healthcare related covariates (see Table~\ref{tab:covariates} in the Supplementary Information).

Following the logic of the counterfactual model, we define the causal effect of a chronic disease on a quality of life indicator as the expected difference between the functional outcome if a subject had, or did not have, the chronic disease. Therefore, we form treatment groups identifying subjects who present chronic conditions at the beginning of the study, and control groups identifying subjects who never develop chronic conditions throughout the study period. This results in 419 treated subjects and 577 control subjects for hypertension.

For each functional outcome, we then fit \FOCAL{}{} using a function-on-scalar neural network specification for the regression function $\hat{\mu}^{(a)}$, and logistic regression for the propensity score $\hat{\pi}^{(a)}$. No time-varying or post-baseline variables beyond the cohort definition are used in the propensity score or outcome models. Results are shown in Figure~\ref{fig:share}. 
In the aggregate, the estimated Functional Average Treatment Effects (FATE), obtained by averaging the difference in pseudo-outcomes computed through \FOCAL{}'s second stage, 
reveal a coherent and clinically consistent narrative regarding the impact of hypertension on quality of life. We observe a progressive deterioration in well-being for treated subjects across both functional outcomes. Specifically, hypertension causes a deepening negative effect on the CASP score over time. Simultaneously, it leads to an increasing positive effect on the mobility index; since higher mobility scores indicate greater physical impairment, this confirms that the condition exerts a detrimental force on both psychological and physical dimensions of health. 

Next, to move beyond aggregate summaries and uncover the drivers of this variation, we turn to a more refined analysis using F-CATE estimates. Specifically, we examine the estimated F-CATE surfaces conditioned on age and gender, as these variables represent critical axes of interest from both scientific and policy perspectives. Moreover, as shown by the distributions in Figure~\ref{fig:share}\textbf{a}, age exhibits the strongest imbalance between treated and control groups, underscoring the importance of the proper correction performed by \FOCAL{}.
The detrimental impact of age on CASP can be observed across the entire temporal domain and intensifies markedly with age, as reported by \citet{anderson1999effect, lloyd2005hypertension}. 
Specifically, the negative effect 
is relatively stable 
across the 192-month observation window for younger individuals, and sharpens over time for older individuals, suggesting a decreased resilience for the latter.
Furthermore, gender appears to play an ``intercept'' role -- amplifying the absolute magnitude of the 
effect of hypertension for male subjects, a result supported by previous clinical evidence \citep{cutler2008trends, hayes1998hypertension, sandberg2012sex, vitale2010gender, everett2015gender}. 
Concerning the mobility index, the estimated F-CATE surface reveals a distinct temporal profile for older individuals, characterized by a moderate mid-period exacerbation followed by a sharp, late-stage increase (a ``bump'' in the estimated surface) during the final stages of the analysis. These
observations align with previous literature 
\citep{burt1995prevalence, ostchega2007trends}. Notably, the role of gender
in this context is to anticipate and exacerbate the profile for women;
older women appear to
experience the adverse consequences of hypertesion on moblity
earlier and more severely than older men \citep{alqahtani2025stepping}.

\subsection{\FOCAL{} sheds light on the heterogeneous role of distributed primary health care in shaping COVID-19 mortality patterns}

To further demonstrate the practical 
value of \FOCAL{}, we use it to investigate the heterogeneous causal effects of distributed primary health care on COVID-19 mortality patterns during the first two pre-vaccine epidemic waves in Italy \citep{boschi2021functional, boschi2023contrasting}. This application showcases \FOCAL{}'s capability to handle complex spatiotemporal domains where mortality patterns are expressed as functional outcomes.

The 
data encompasses all 107 Italian provinces throughout the two
major pre-vaccine waves of the pandemic. We use as observation windows the 150 days from February 25, 2020 to July 23, 2020 for the first wave, and the 150 days from October 1, 2020 to February 27, 2021 for the second. The first wave 
was characterized by dramatic, spatially concentrated mortality peaks, particularly in northern regions like Lombardia. In contrast, the second wave
was less dramatic, more widespread and asynchronous across the country. The
main goal of this analysis is to evaluate the heterogeneous causal effects of distributed primary health care on province-level COVID-19 mortality rate trajectories. 
We define the treatment using records of the number of adults per family doctor (see Methods for details); ``treated'' provinces are those with higher ratios, i.e. fewer family doctors, and thus poorer distributed primary health care.

The analysis incorporates several critical socio-demographic and environmental covariates that may confound the relationship between mortality and 
the treatment, including the percentage of the population over 65 years of age; average beds per hospital, 
students per classroom and employee per firm (proxies for the ability of hospitals, schools and workplaces to act as contagion hubs), and PM10 pollution levels. 
In addition, following the logic of \citet{boschi2023contrasting},
we consider the covariate ``area before'' -- a province-specific standardized metric evaluating the area under the mortality curve up to the point when 
restriction measures were imposed nationally or locally. This standardized cumulative mortality captures the degree to which the epidemic was able to ``build up'' within a province prior to the implementation of lockdowns and social distancing -- avoiding the use of case counts, which were highly unreliable. As shown in 
Figure~\ref{fig:covid}\textbf{a}, several of 
the covariates under consideration exhibit 
marked imbalances between ``treated'' and ``untreated'' provinces, thus necessitating the 
adjustment provided by \FOCAL{}.

\begin{figure}[t!]
    \centering
    \includegraphics[width=\linewidth]{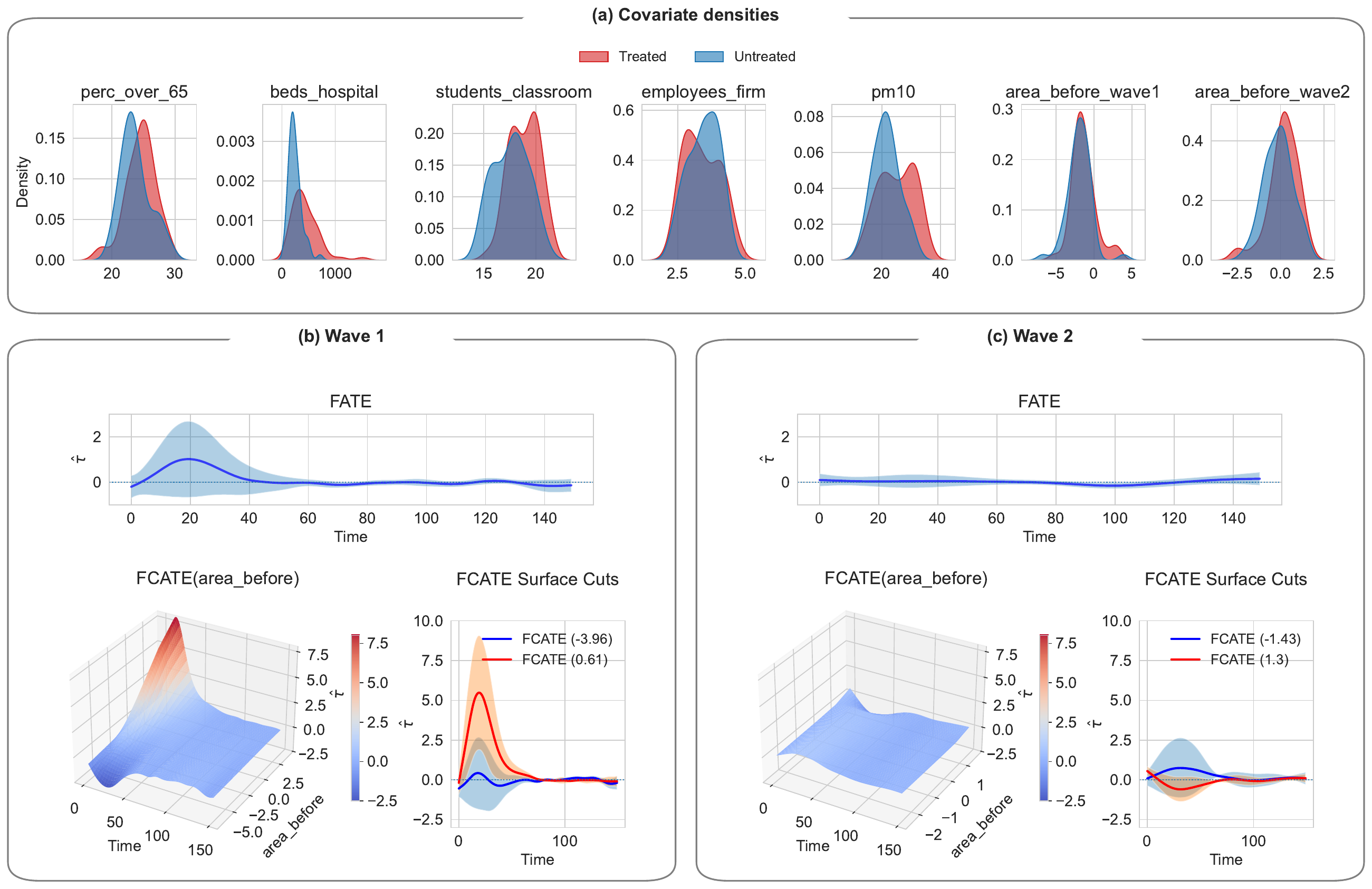}
    \caption{Heterogeneous impact of 
    availability of primary health care on COVID-19 mortality patterns (COVID-19 study). \textbf{(a)} 
    Distribution of key socio-demographic and environmental baseline covariates for provinces with high (``Treated'') versus low (``Untreated'') 
    number of adults per primary care physician (treated provinces are those with poorer distributed primary health care).
    Sizeable imbalances between groups
    are observed in most
    of the covariates,
    motivating the propensity adjustment performed 
    by \FOCAL{}.
    \textbf{(b)} Wave 1 results. Top: The estimated Functional Average Treatment Effect (FATE) for the first wave (February-July 2020)
    suggests 
    that poor distributed primary health care exacerbated mortality 
    in the initial stages,
    but this is non-significant based on the 95\% confidence band. 
    Bottom left: The estimated F-CATE surface conditioned on the ``area before'' standardized metric disaggregates this effect, revealing a dramatic
    escalation in provinces with the highest early outbreak intensity, and thus pinpointing how poor distributed primary health care was most consequential where the epidemic
    had the strongest early momentum.
    Bottom right: Cuts of the 
    surface plot at 
    the 10\textsuperscript{th} (blue) and 90\textsuperscript{th} (red) percentiles of
    the marginal distribution of ``area before'', 
    with 90\% confidence bands supporting strong significance at the latter. \textbf{(c)} Wave 2 results. Top: The estimated FATE for the second wave (October 2020–February 2021)
    appears rather flat with a tight 95\% confidence band around $0$, 
    suggesting a negligible impact of distributed primary health care during the more widespread and asynchronous epidemic unfolding that characterized
    this wave compared to the first.
    Bottom left: The estimated F-CATE surface confirms a negligible impact also when conditioning on ``area before''. 
    Bottom right: Cuts of the
    surface plot at 
    the 10\textsuperscript{th} (blue) and 90\textsuperscript{th} (red) percentiles of the marginal distribution of ``area before'',
    again with 90\% bands. 
    }
    \label{fig:covid}
\end{figure}

The estimated Functional Average Treatment Effects (FATE), obtained by averaging the difference in pseudo-outcomes computed through \FOCAL{}'s second stage, fail to provide a 
meaningful 
picture of the
impacts of distributed primary health care on mortality. In fact, in both the first and second 
wave (Figure \ref{fig:covid}\textbf{b} and \textbf{c}, top), the estimated FATE trajectories are not statistically significant.
Thus, marginal averages 
lead to the conclusion that the availability of family doctors
did not have a role in curbing mortality. 

However, this lack of statistical significance 
masks heterogeneous causal impacts. To
disaggregate such impacts, we turn to the F-CATE surfaces
conditioning on
``area before'' -- a standardized metric of early epidemic momentum. By focusing specifically on the interaction between 
``area before'' and the temporal progression of mortality, \FOCAL{} reveals a striking and localized story. During the first wave, the estimated F-CATE surface (Figure \ref{fig:covid}\textbf{b}, bottom)
shows a dramatic and significant peak in the treatment effect for provinces with high early epidemic momentum.
This 
pinpoints poor distributed primary health care 
as a primary driver of mortality exactly in those provinces where COVID-19 had established a strong foothold before the lockdowns. In these provinces, family doctors made a decisive difference as the critical first line of defense, particularly at a stage when standardized clinical protocols and dedicated assistance infrastructures were not yet in place. In contrast, 
the estimated F-CATE surface for the second wave (Figure \ref{fig:covid}\textbf{c}, bottom), in agreement with the aggregate effect, 
reflects a
weakening and flattening of these dynamics. 
The relative importance of primary care in curbing mortality diminished as the epidemic became more geographically widespread and asynchronous, and as the broader healthcare system integrated more mature and diverse response mechanisms. These findings illustrate 
\FOCAL{}'s ability to provide
granular, spatiotemporal insights -- 
unveiling how local conditions and the timing of the outbreak modulated the role of distributed primary health care.

\section{Discussion}
In this paper, we introduce \FOCAL{} (Functional Outcome Causal Learning), a novel doubly robust meta-learner designed to estimate heterogeneous, covariate-dependent treatment effects when outcomes are expressed as functions over a continuous domain. Our proposal addresses a critical gap in the causal inference literature, where existing meta-learning approaches are largely confined to scalar outcomes, thereby failing to fully leverage the rich, continuous information inherent in functional data. By integrating advanced machine learning and functional data analysis tools, \FOCAL{} 
advances the capabilities of machine intelligence to provide nuanced and individualized causal insights.

A key strength of \FOCAL{} lies in its \textit{doubly robust} design, a property inherited and extended from the well-established DR-Learner framework \citep{kennedy2023towards}. This ensures that \FOCAL{} estimates of the F-CATE remain asymptotically consistent even if one of the models employed for the nuisance functions (functional outcome regressions and propensity score) is misspecified, provided the other is correctly specified. This robustness is paramount for real-world machine intelligence applications, where achieving perfect model specification is often challenging due to lack of prior knowledge, data complexity or inherent biases. Furthermore, \FOCAL{}'s flexibility to incorporate arbitrary, state-of-the-art machine learning algorithms for modeling the nuisance functions and the final F-CATE regression allows it to capture complex, non-linear relationships without imposing rigid parametric assumptions. This is crucial for capturing how the entire functional response varies across different subpopulations, and thus provide granular and interpretable causal insights that go beyond simple scalar averages. 

We confirmed \FOCAL{}'s superior performance and robustness against existing non-robust functional methods through a comprehensive simulation study, and we demonstrated its practical scope through applications to
real-world functional datasets
from the SHARE study \citep{mannheim2005survey}
and from the COVID-19 epidemic in Italy \citep{boschi2023contrasting}. 
Also importantly, \FOCAL{} possesses theoretical guarantees, including an oracle property and valid simultaneous confidence bands (see Methods and Supplementary Information for details). 
Such guarantees further solidify \FOCAL{}'s statistical rigor and enable reliable uncertainty quantification for the estimated F-CATE.

Despite its strengths, \FOCAL{} operates under a number of assumptions that are typical in the causal inference and functional data analysis literature.
We note that, albeit standard, these assumptions are crucial; their violation could impact the  effectiveness of \FOCAL{}, as well as the validity of causal claims. Moreover, while \FOCAL{}'s use of highly flexible machine learning algorithms is a strength, it can also lead to increased computational
burden. 
For the simulations and applications presented in this paper, computational costs were negligible (in the order of minutes on standard laptops). However, they may escalate in applications with much larger sample sizes or much more densely sampled functional data.
Future work could explore more computationally efficient functional regression techniques or distributed computing strategies to enhance scalability.

Looking ahead, our proposal opens several exciting avenues for future research. Extensions could include adapting \FOCAL{} to handle multimodal covariates, which are increasingly common in biomedical and social science applications \citep{boschi2024functional, ektefaie2023multimodal, ngiam2011multimodal}. Foundation models may be employed to learn latent embeddings, simultaneously leveraging data with different modalities (curves, images, scalars, labels, etc.) -- and thus providing further insights on heterogeneous effects \citep{rajendran2024learning, bengio2013representation}. Methods for estimating the effects of time-varying treatments and dynamic treatment regimes on functional outcomes would be another important research direction \citep{tan2025causal}. Furthermore, exploring the integration of cutting-edge deep learning architectures specifically designed for functional data within \FOCAL{}'s nuisance function and F-CATE estimation stages could unlock even greater flexibility and predictive power \citep{rao2023nonlinear, rao2023modern}. 

Ultimately, \FOCAL{} provides a rigorous foundation for applying modern causal learning techniques to the rich data structures increasingly encountered in scientific research, paving the way for more sophisticated and trustworthy AI systems capable of nuanced causal understanding and precision decision-making in personalized health, adaptive policy design, and beyond.

\section{Methods}

\subsection{Notation and target identification}
Let $\mathcal T = [t_1,t_2]$ be a closed bounded interval. Without loss of generality, we shall consider $\mathcal T = [0,1]$. Let $A\in\{0,1\}$ be a binary variable that indicates whether a subject belongs to the treated group ($A=1$) or to the control group ($A=0$); $X\in\RR^p$ be a $p$-dimensional vector of pretreatment covariates; and $\mathcal{Y}^{(0)}, \mathcal{Y}^{(1)}\in \Ltwo (\mathcal{T})$ be square integrable functions defined over $\mathcal{T}$, representing the potential outcomes under control and treatment, respectively. A prototypical \textit{full-data} sample is $\data^F = \left(X, A, \mathcal{Y}^{(0)}, \mathcal{Y}^{(1)}\right)$. Of course, in practice, we never have access to $\data^F$, as we can only observe one of the two potential outcomes. More formally, we have access to a collection of $n$ independent and identically distributed \textit{observed-data} samples, $\left\{\data_i = (A_i, X_i, \mathcal{Y}_i)\right\}_{i=1}^n$, where $\mathcal{Y}_i = A_i\mathcal{Y}^{(1)}_i + (1-A_i)\mathcal{Y}^{(0)}_i$ represents the observed outcome. We let $\data = (A,X,\mathcal{Y})$ denote an independent copy of $\mathcal{D}_i = (A_i,X_i,\mathcal{Y}_i)$. Finally, we assume that $\Ltwo(\mathcal{T})$, the space of square integrable functions over $\mathcal{T}$, is equipped with the $\Ltwo$-norm defined by $\norm{f}^2=\int_{\mathcal{T}}(f(t))^2\,dt<\infty$, for $f\in \Ltwo(\mathcal{T})$, thus making $(\Ltwo(\mathcal{T}), \norm{\cdot})$ a Hilbert space. In principle, $\mathcal T$ could be multidimensional, but for simplicity we consider it to be a proper subset of $\RR$. 

Our target is the \textit{functional conditional average treatment effect} (F-CATE), defined as
\begin{equation}
\label{eq:FCATE}
    \thetatarget(x) = \EE{\mathcal{Y}^{(1)}- \mathcal{Y}^{(0)}\mid X=x}\,,
\end{equation}
where the expectation is taken over the data generating process. Recall that the potential outcomes $\mathcal{Y}^{(0)}, \mathcal{Y}^{(1)}$ are functions over $\mathcal{T}$; so is $\thetatarget(x)$ for each $x\in\RR^p$ (we omit arguments over $\mathcal{T}$ for notational simplicity). To provide some intuition, if $p=1$ and thus $x\in\RR$, $\thetatarget$ is a surface defined over the domain $\mathcal{T} \times \RR$, and each evaluation of $\thetatarget$ at $x$, i.e.~$\thetatarget(x)$, is a function defined over $\mathcal{T}$.

The F-CATE target in Equation~\ref{eq:FCATE} is a function of the full data, and as such it is not directly learnable from the observed data. With the following minimal set of assumptions, which are standard in the causal inference literature \citep{kennedy2023towards, wager2024causal}, we can identify the target relying exclusively on the observed data. 
\begin{assumption}[Identifiability]
\label{ass:identify}
Let the following hold:
    \begin{enumerate}[label=\textbf{\alph*.}]
    \item \textbf{Consistency.} The potential outcome under a specific treatment is the same regardless of the mechanism by which the treatment is administered; that is, $\mathcal{Y} = \mathcal{Y}^{(a)}$ if $A=a$. Equivalently, $\mathcal{Y} = A\mathcal{Y}^{(1)} + (1-A)\mathcal{Y}^{(0)}$.
    \item \textbf{No unmeasured confounding.} 
    $(\mathcal{Y}^{(0)}, \mathcal{Y}^{(1)}) \indep A \,|\, X$.
    \item \textbf{Weak positivity.} $0<\PP{A=1\,|\,X=x}<1$ for every $x\in\RR^p$ almost surely.
\end{enumerate}
\end{assumption}
Define the \textit{regression functions} as
\begin{equation}
   {\mutarget}^{(a)}(x) = \EE{\mathcal{Y}\,|\,X=x,\,A=a}\,,\quad a\in\{0,1\}\,,
\end{equation}
and the \textit{propensity score function} as 
\begin{equation}
   \pitarget(x) = \PP{A = 1\,|\,X=x} = \EE{\one\{A = 1\}\,|\,X = x}\,.
\end{equation}
These are usually referred to as {\em nuisance functions}; note that they depend only on the observed data \citep{kennedy2022semiparametric}. Finally, define the \textit{pseudo-outcomes} as
\begin{equation}
    {\gammatarget}^{(1)}(\data) = {\mutarget}^{(1)}(X) + \frac{A}{\pitarget(X)} \left(\mathcal{Y} - {\mutarget}^{(1)}(X) \right)\,,\quad {\gammatarget}^{(0)}(\data) = {\mutarget}^{(0)}(X) + \frac{1-A}{1-\pitarget(X)} \left(\mathcal{Y} - {\mutarget}^{(0)}(X) \right)\,.
\end{equation}

The following claim makes the identification of F-CATE explicit. 
\begin{lemma}
\label{th:equivalence}
Under Assumption~\ref{ass:identify} (identifiability), the F-CATE target defined in Eq.~\ref{eq:FCATE} can be rewritten as:
\begin{equation}
    \thetatarget(x) = \EE{{\gammatarget}^{(1)}(\data)- {\gammatarget}^{(0)}(\data) \mid X=x}\,.
\end{equation}
\end{lemma}
We defer the proof of this and subsequent claims to the Supplementary Information, and are now ready to provide full detail on the proposed \FOCAL{}.

\subsection{\FOCAL{} in detail}
\FOCAL{} is constructed as a doubly robust meta-learner for $\thetatarget(x)$, building upon the principles of the DR-Learner \citep{kennedy2023towards} but explicitly extending them to functional outcomes. \FOCAL{}'s pipeline comprises three stages: (i) estimation of the nuisance functions, (ii) reconstruction of the functional pseudo-outcomes, and (iii) final estimation of the F-CATE.

In the {\em first stage}, \FOCAL{} estimates the following nuisance functions:
\begin{itemize}
    \item \textbf{Functional outcome regressions:} these represent the conditional expectations of the functional outcome given the covariates and the treatment status, denoted as ${\mutarget}^{(a)}(x) = \EE{\mathcal{Y}\mid X=x, A=a}$, $a \in \{0,1\}$. They are estimated fitting functional regression models to the observed data \citep{kokoszka2017introduction}. We note that other machine learning models, such as specialized neural networks for functional data \citep{boschi2024fungcn, yao2021deep}, can be employed as base learners for this task. The estimates are denoted as $\hat\mu^{(a)}(x)$, $a\in\{0,1\}$.
    \item \textbf{Propensity score:} this represents the probability of receiving treatment given the covariates, denoted as $\pitarget(x) = \PP{A=1\mid X=x}$. Its estimation can employ a variety of models and algorithms for binary classification (e.g., logistic regression, random forests, etc.). The estimate is denoted as $\hat\pi(x)$.
\end{itemize}

In the {\em second stage}, \FOCAL{} reconstructs the functional pseudo-outcomes $\gamma^{(a)}(\data_i)$, $a\in\{0,1\}$ for each individual $i=1,\ldots,n$. These capture the individual-level treatment effect on the functional outcome. For an individual with observed outcome $\mathcal Y_i$, covariates $X_i$ and treatment $A_i$ the estimates are: 
\begin{equation}
    \hat\gamma^{(1)}(\data_i) = \hat\mu^{(1)}(X_i) + \frac{A_i}{\hat\pi(X_i)}\left(\mathcal{Y}_i - \hat\mu^{(1)}(X_i) \right)\,,\quad \hat\gamma^{(0)}(\data_i) = \hat\mu^{(0)}(X_i) + \frac{1-A_i}{1-\hat\pi(X_i)}\left(\mathcal{Y}_i - \hat\mu^{(0)}(X_i) \right)\,.
\end{equation}
A fundamental property of this construction, extended to the functional setting, is its double robustness. Double robustness is critical for the reliability of \FOCAL{} estimates, as consistency is maintained even if one of the nuisance models is misspecified, which is a common occurrence in real-world applications. We note that, in applications involving functional outcomes, it would be reasonable to assume that estimation of the regression functions ${\mutarget}^{(a)}(x)$, $a\in\{0,1\}$, is a more complex task than estimation of the propensity score function $\pitarget(x)$. Indeed, ${\mutarget}^{(a)}(x)$ maps a vector of size $p$ into a function in $C(\mathcal{T})$, while $\pitarget(x)$ maps the same vector into a variable in $(0,1)$. Double robustness ensures that, if $\hat{\pi}(x)$ converges to $\pitarget(x)$, $\hat\gamma^{(a)}(\data)$ converges to ${\gammatarget}^{(a)}(\data)$ for \textit{any} choice of $\hat{\mu}^{(a)}(x)$. See \citet{kennedy2022semiparametric, testa2025doubly} for additional information and perspectives on double robustness.

In the {\em third stage}, \FOCAL{} estimates the F-CATE $\thetatarget(x)$ by fitting a final regression model for the difference between the functional pseudo-outcomes $\hat\gamma^{(1)}(\data) - \hat\gamma^{(0)}(\data)$ (a functional response) against the covariates $X$. This involves again utilizing functional regression techniques. The resulting $\hat\theta(x)$ provides a comprehensive and interpretable map of how the causal effect on the functional outcome varies across the entire covariate space.

To mitigate overfitting of the nuisance functions and ensure valid asymptotic inference for the \FOCAL{} estimator, we employ \textit{cross-fitting} \citep{ahrens2025introduction,chernozhukov2018double}. We randomly split the observations $\{\data_1,\dots,\data_{n}\}$ into $J$ disjoint folds (without loss of generality, assume that $n$ is divisible by $J$). For each $j=1,\ldots, J$ we form $\hat{\mathbbmss{P}}^{[-j]}$ 
with all but the $j$-th fold, and $\mathbbmss{P}_{n}^{[j]}$ with the $j$-th fold. Then, we learn $\hat{\mu}^{(a)[-j]}(x)$, $a \in \{0,1\}$ and
$\hat{\pi}^{[-j]}(x)$ using data in $\hat{\mathbbmss{P}}^{[-j]}$, and compute $\hat{\gamma}_{\hat\mu^{(a)[-j]};\hat\pi^{[-j]}}^{(a)}(\data)$, $a \in \{0,1\}$ and regress their difference on $X$ using data in $\mathbbmss{P}_n^{[j]}$. This produces a fold-specific estimate $\hat\theta^{[j]}(x)$ of the F-CATE. Finally, we average across folds to obtain our cross-fitted \FOCAL{} estimator as
\begin{equation}
    \hat\theta(x) = \frac{1}{J} \sum_{j=1}^J \hat\theta^{[j]}(x) \,. 
\end{equation}
Cross-fitting ensures that the nuisance function estimation underlying the reconstruction of functional pseudo-outcomes for any given observation is always performed on an independent subset of the data. The functional pseudo-outcomes are computed for all observations by iterating this process across all $J$ folds. This orthogonalization effectively debiases the final F-CATE estimation and allows for full sample efficiency, even when highly flexible machine learning models are used for nuisance function estimation. 

\subsection{Theoretical guarantees and inference}
In addition to 
double robustness,
\FOCAL{} 
enjoys an oracle property if the regression of the difference between functional pseudo-outcomes on the covariates is \textit{stable}, as defined below.
\begin{definition}[Stability]
    Let $\hat{\mathbbmss{P}}^{[-j]}$ and $\mathbbmss{P}^{[j]}_n$ be two disjoint sub-samples obtained by randomly splitting observations with cross-fitting. Let:
    \begin{itemize}
    \item $\hat{f}^{[-j]}(\data)$ be an estimate of $f^\star(\data)$ based on data in $\hat{\mathbbmss{P}}^{[-j]}$.
    \item $\hat{b}^{[-j]}(x) = \EE{\hat{f}^{[-j]}(\data)-f^\star(\data)\mid \hat{\mathbbmss{P}}^{[-j]}, X=x}$ be the pointwise conditional bias of the estimator $\hat{f}^{[-j]}$.
    \item $\empEE{[j]}{\cdot\mid X=x}$
    be a generic estimator of $\EE{\cdot\mid X=x}$ based on data in ${\mathbbmss{P}}^{[j]}_n$.
    \end{itemize}
    The estimator $\empEE{[j]}{\cdot\mid X=x}$ is \textit{stable} at $X=x$ with respect to a distance metric $d$ if
    \begin{equation}
        \frac{\norm{\empEE{[j]}{\hat f(\data)\mid X=x} - \empEE{[j]}{f^\star(\data)\mid X=x} - \empEE{[j]}{\hat b^{[-j]}(X)\mid X=x}}}{\norm{\sqrt{\EE{\left(\empEE{[j]}{f^\star(\data)\mid X=x} - \EE{f^\star(\data)\mid X=x}\right)^2}}}} \pto 0\,,
    \end{equation}
    whenever $d\left(\hat f, f^\star\right)\pto 0$.
\end{definition}
The definition of stability of an estimator was originally introduced by \citet{kennedy2023towards} for one-dimensional outcomes;
our definition extends it
accounting for the infinite-dimensional nature of functional outcomes. Endowed with 
this definition, we
now state the
oracle 
property 
in the following theorem.
\begin{theorem}[Oracle property]
\label{th:oracle}
     Let $\hat{\mathbbmss{P}}^{[-j]}$ and $\mathbbmss{P}^{[j]}_n$ be two disjoint sub-samples obtained by randomly splitting observations with cross-fitting. Assume that:
    \begin{itemize}
    \item The
    estimator $\hat\theta^{[j]}(x) = \empEE{[j]}{\hat\gamma^{(1)}_{\hat\mu^{(1)[-j]};\hat\pi^{[-j]}}(\data) - \hat\gamma^{(0)}_{\hat\mu^{(0)[-j]};\hat\pi^{[-j]}}(\data) \mid X=x}$ is stable with respect to a distance $d$.
    \item $d\left(\hat\gamma^{(a)}_{\hat\mu^{(a)[-j]};\hat\pi^{[-j]}}(\data),{\gammatarget}^{(a)}(\data)\right)\pto 0$ for $a\in\{0,1\}$. 
    \end{itemize}
    Let $\Tilde{\theta}^{[j]}(x) = \empEE{[j]}{{\gammatarget}^{(1)}(\data) - {\gammatarget}^{(0)}(\data)\mid X=x}$ 
    be an oracle estimator that regresses the true difference in pseudo-outcomes on the covariates
    using data in $\mathbbmss{P}^{[j]}_n$, and denote its risk by
    \begin{equation}
        \Tilde{R}^{[j]}_n(x) = \norm{\sqrt{\EE{\left(\empEE{[j]}{{\gammatarget}^{(1)}(\data) - {\gammatarget}^{(0)}(\data)\mid X=x} - \EE{{\gammatarget}^{(1)}(\data) - {\gammatarget}^{(0)}(\data)\mid X=x}\right)^2}}}\,.
    \end{equation}
    Then 
    \begin{equation}
        \norm{\hat\theta^{[j]}(x) - \Tilde{\theta}^{[j]}(x)} = \norm{\empEE{[j]}{\hat b^{[-j]}(X)\mid X=x}} + o_\mathbbmss{P}\left(\Tilde{R}^{[j]}_n(x)\right)\,,
    \end{equation}
    where 
    \begin{equation}
        \hat b^{[-j]}(x) = \frac{\left(\hat\pi^{[-j]}(x) - \pitarget(x)\right)\left(\hat\mu^{(1)[-j]}(x) - {\mutarget}^{[1]}(x)\right)}{\hat\pi(x)} + \frac{\left(\hat\pi^{[-j]}(x) - \pitarget(x)\right)\left(\hat\mu^{(0)[-j]}(x) - {\mutarget}^{[0]}(x)\right)}{1- \hat\pi(x)}\,.
    \end{equation}
    If it further holds that $\norm{\empEE{[j]}{\hat b^{[-j]}(X)\mid X=x}} = o_\mathbbmss{P}\left(\Tilde{R}^{[j]}_n(x)\right)$, then $\hat\theta^{[j]}(x)$ is oracle efficient, i.e.~asymptotically equivalent to the oracle estimator $\Tilde{\theta}^{[j]}(x)$ in the sense that
    \begin{equation}
        \frac{\norm{\hat\theta^{[j]}(x) - \Tilde{\theta}^{[j]}(x)}}{\Tilde{R}^{[j]}_n(x)} \pto 0\,.
    \end{equation}
\end{theorem}
The previous result is powerful, as it provides strong asymptotic estimation guarantees for \FOCAL{}. Moreover, it has actionable consequences for inferential purposes,
as articulated in the following proposition.
\begin{proposition}[Valid and simultaneous coverage]
    \label{th:inf}
    Let $\hat{\mathbbmss{P}}^{[-j]}$ and $\mathbbmss{P}^{[j]}_n$ be two disjoint sub-samples obtained by randomly splitting observations with cross-fitting. Assume:
    \begin{itemize}
    \item The 
    estimator $\hat\theta^{[j]}(x)$ is stable with respect to a distance $d$.
    \item $d\left(\hat\gamma^{(a)}_{\hat\mu^{(a)[-j]};\hat\pi^{[-j]}}(\data),{\gammatarget}^{(a)}(\data)\right)\pto 0$ for $a\in\{0,1\}$. 
    \item The oracle estimator $\Tilde{\theta}^{[j]}(x)$ has an asymptotic Gaussian process behavior, that is
        \begin{equation}
        \sqrt{n}\left(\Tilde{\theta}^{[j]}(x) - \thetatarget(x) \right) \dto \mathcal{GP}\left(0, \Sigma(x)\right)\,,
        \end{equation}
        where $\Sigma(x) = \VV{{\gammatarget}^{(1)}(\data) - {\gammatarget}^{(0)}(\data) \mid X=x}$.
    \end{itemize}
    Then the asymptotic distribution of $\hat{\theta}^{[j]}(x)$ matches the one of $\Tilde{\theta}^{[j]}(x)$ with $\hat\Sigma(x)$, where we define $\hat\Sigma(x) = \empEE{[j]}{\left({\hat\gamma}^{(1)}_{\hat\mu^{(1)[-j]};\hat\pi^{[-j]}}(\data) - {\hat\gamma}^{(0)}_{\hat\mu^{(0)[-j]};\hat\pi^{[-j]}}(\data)\right) \left({\hat\gamma}^{(1)}_{\hat\mu^{(1)[-j]};\hat\pi^{[-j]}}(\data) - {\hat\gamma}^{(0)}_{\hat\mu^{(0)[-j]};\hat\pi^{[-j]}}(\data)\right)^T \mid X=x}$.
\end{proposition}
Thus, for any
given covariate profile $X=x$, we can provide valid and simultaneous confidence bands through a parametric bootstrap, as proposed by \citet{pini2017interval}. We repeatedly sample from the Gaussian process
and estimate quantiles
-- without any additional assumption on the covariance function. In symbols, the 
band takes the form $C_\alpha(x) = [\hat{\theta}_{\alpha/2}(x), \hat{\theta}_{1-\alpha/2}(x)]$, where $\hat{\theta}_{\alpha/2}(x)$ and $\hat{\theta}_{1-\alpha/2}(x)$ indicate the estimated $\alpha/2$ and $1-\alpha/2$ quantiles, respectively. Given its ease of implementation, and the fact that it requires no additional assumptions, we adopt
this parametric bootstrap approach for evaluating uncertainty in \FOCAL{} results.

\subsection{Simulation study}
To rigorously assess the performance of \FOCAL{}  in synthetic scenarios, we design a simulation experiment involving $n=5000$ independent observations with functional outcomes observed on a dense grid of $100$ equidistant 
points over the domain $\mathcal{T}=[0,1]$.

We first generate four baseline covariates $X_1, X_2, X_3, X_4$ from a standard multivariate normal distribution $\Normal{0}{I_4}$. Treatment assignment $A_i \in \{0,1\}$ is drawn from a Bernoulli distribution with propensity scores $\pitarget(X_1, X_2, X_3, X_4) = \text{logit}^{-1}(\eta(X_1, X_2, X_3, X_4))$, where the inverse logit function is defined as $\text{logit}^{-1}(x) = 1/(1+\exp{(-x)})$ and the linear function $\eta(X_1, X_2, X_3, X_4)= -X_1 + 0.5X_2 - 0.25X_3 - 0.1X_4$ creates a strong dependency on the baseline covariates. The functional regression functions are generated using a time-varying coefficient model, that is:
\begin{equation}
    {\mutarget}^{(a)}(X_1,X_2,X_3,X_4) =  \beta_0 + \sum_{j=1}^4 \beta_jX_j + \delta(a)\beta_5X_1 \,,
\end{equation}
where $\delta(a)=1$ if $a=1$ and 0 otherwise, and the $\beta_j$'s depend on $t \in \cal{T}$. Specifically, to ensure realistic smoothness and correlation structures, the coefficients $\beta_j$ are sampled from a Gaussian Process with Matérn covariance kernel defined, for any $t,s\in\mathcal{T}$, as:
\begin{equation}
\label{eq:matern}
        C(s,t) = \frac{\alpha^2}{\Gamma(\nu)2^{\nu-1}}\left(\frac{\sqrt{2\nu}}{l} \lvert s-t \rvert\right)^\nu \chi_\nu\left(\frac{\sqrt{2\nu}}{l} \lvert s-t \rvert \right)\,.
\end{equation}
Here $\chi_\nu$ is a modified Bessel function, and we set the parameters to $l=0.25$, $\nu=5.5$, $\alpha=2$ for $\beta_1,\dots,\beta_4$ and $\alpha=10$ for $\beta_5$. Then, for each observation $i=1,\dots,n$, we 
define potential outcomes as 
\begin{equation}
    \mathcal{Y}^{(a)}_i =  {\mutarget}^{(a)}(X_{i1},X_{i2},X_{i3},X_{i4}) + \varepsilon_i
    \,,\quad a\in\{0,1\}\,,
\end{equation}
where the noise term $\varepsilon$ is drawn from a Gaussian Process with Matérn covariance kernel ($l=0.25$, $\nu=5.5$, $\alpha=10$). Finally, for each observation $i=1,\dots,n$, the observed functional outcome is
\begin{equation}
    \mathcal{Y}_i = A_i \mathcal{Y}^{(1)}_i + (1-A_i) \mathcal{Y}^{(0)}_i\,.
\end{equation}
This construction induces a heterogeneous treatment effect driven by $X_1$ and $\beta_5$, resulting in a non-trivial F-CATE target $\thetatarget(x)$.

To control the degree of model misspecification, we also define the non-linear transformations $Z_1,Z_2,Z_3,Z_4$ of the covariates $X_1, X_2, X_3, X_4$ defined as 
\begin{equation}
Z_1 = \exp(X_1/2)\,,\quad
Z_2 = \frac{X_2}{1+\exp(X_1)} + 10\,,\quad
Z_3 = \biggl(\frac{X_1 X_3}{25}+0.6\biggr)^3\,,\quad
Z_4 = (X_2 + X_4 + 20)^2\,.
\end{equation}
An analyst who fits a linear model for the outcome or a logistic model for the propensity score using the covariates $Z_1,Z_2,Z_3,Z_4$ will have a misspecified model. A correct specification would require the analyst to know the true latent variables $X_1, X_2, X_3, X_4$ or the exact inverse transformations. In all cases, 5-fold cross-fitting is used. We thus consider four scenarios:
\begin{enumerate}
    \item Correct specification: both nuisance models are correctly specified;
    \item Misspecified regression function: only $\hat\mu^{(a)}$ is misspecified by injecting the non-linear transformation through $Z_1,Z_2,Z_3,Z_4$;
    \item Misspecified propensity score: only $\hat\pi$ is misspecified 
    by injecting the non-linear transformation through $Z_1,Z_2,Z_3,Z_4$;
    \item Fully misspecified: $\hat\mu^{(a)}$ and $\hat\pi$ are both misspecified.
\end{enumerate}

This design allows us to empirically assess the double robustness property of \FOCAL{}: consistency is preserved if at least one nuisance model is correctly specified, and deteriorates only under joint misspecification. 

We repeat the simulation experiment described above 5000 times. For each scenario, we compute the aggregated root mean squared error (ARMSE) between the estimated functional treatment effect $\hat{\theta}$ and the ground truth $\thetatarget(x)$, where~$\text{ARMSE}(\hat\theta,\thetatarget) = \int_x\norm{\hat\theta(x) - \thetatarget(x)}\,dx$. This metric provides a direct and interpretable measure of estimation accuracy across the covariate space, as well as time through the $\mathbbmss{L}_2$ norm.

\subsection{SHARE study}
SHARE ({\em Survey of Health, Aging and Retirement in Europe}) is a research infrastructure that aims to investigate the effects of health, social, economic
and environmental policies on the life course of European citizens \citep{borsch2013data, bergmann2017survey, borsch2020survey}. 
We preprocess data following the steps described in \citet{boschi2024new}. We focus on the 1518 subjects who participated in at least seven out of the eight waves (i.e., survey times).
We investigate a subset of the variables from the EasySHARE dataset \citep{gruber2014generating}, a preprocessed version of the SHARE data. While some of these variables are characterized by values that change over time (e.g., CASP and mobility index) and are suitable for a functional representation, others are scalar (e.g., education years) or categorical (e.g., gender) and do not evolve across waves. 
We smooth time-varying variables using cubic \emph{B-splines} with knots at each survey date and roughness penalty on the curves second derivative \citep{ramsay2005}. For each curve, the smoothing parameter is selected by minimizing the average generalized cross-validation error \citep{craven1978smoothing}. 
Note that, although survey dates and number of measurements may vary across subjects, creating a functional representation provides a natural imputation for missing values and facilitates the comparison of different statistical units across the entire temporal domain.

For this application, \FOCAL{} 
employs neural networks for 
the regression nuisance function and the pseudo-outcome regression, 
and a logistic regression for the propensity score nuisance function. For the regression nuisance function, we opt for a two-layer neural network (10 neurons per layer) with adaptive learning rate and 50k iterations.
For the pseudo-outcome regression
we opt for a simpler, one-layer neural network with 5 neurons to encourage smoothness. A small Ridge penalization ($1e^{-3}$) is added for regularization. We employ cross-fitting with 5 balanced folds.

\subsection{COVID-19 study}
This study 
investigates mortality data at the provincial level for the two pre-vacine COVID 19 waves that occurred in Italy 
in 2020-2021. Specifically, we consider differential mortality for 107 provinces, computed for a total of 300 longitudinal points ($150$ days for each wave), as provided by \citet{boschi2023contrasting}. In addition to mortality, we consider  
a set of covariates
describing some key
socio-economic, demographic and infrastructural indicators for each province.
To define our treatment $A$, we binarize 
one such covariate; namely, the ratio between adults and family doctors, as observed for each
province in 2019.
Using the national median as a cut point, we consider as ``treated'' provinces above the median (poor distributed primary health care; $A=1$) and as ``untreated'' provinces below the median (better distributed primary health care; $A=0$).

Estimates of the nuisance functions and the pseudo-outcome regression are obtained as in the
SHARE 
application (neural networks for the regression nuisance function and the pseudo-outcome regression, and logistic regression for the propensity score nuisance function). Here though both neural networks are three-layered -- the former with 20 neurons per layer and the latter with 10 neurons per layer, both with ``relu'' activation function and Ridge penalization. Again, \FOCAL{} is run via a 5-fold cross-fit.

To obtain FATE estimates, we take the average over the pseudo-outcome curves, trimming the top 1\% observations showing the largest norms. To compute $\hat\Sigma$ as in Proposition \ref{th:inf}, which is required to build confidence bands of F-CATE surface cuts, we employ a multi-output random forest regression model (100 trees, 200 minimum samples per leaf). Evaluating the fitted model at any specific $x$ yields $\hat{\Sigma}(x)$.

\section*{Data availability}
Due to data-sharing restrictions, we cannot release the data used in our SHARE application. However, access to data from the SHARE study can be freely requested through the official portal at \url{https://share-eric.eu/data}. Data
for our COVID-19 application are available at \url{https://github.com/tobiaboschi/fdaCOVID2}. 

\section*{Code availability}
The implementation code for \FOCAL{}, as well as the code needed to reproduce the simulation study and the applications presented in this paper, is available on GitHub at \url{https://github.com/testalorenzo/FOCaL}.

\section*{Acknowledgments}
L.T.~wishes to thank Edward Kennedy and the Causal Inference Reading Group at Carnegie Mellon for helpful discussions. The work of F.C.~was partially supported by the Huck Institutes of the Life Sciences at Penn State, the L'EMbeDS Department of Excellence of the Sant'Anna School of Advanced Studies, and the SMaRT COnSTRUCT project (CUP J53C24001460006, as part of FAIR, PE0000013, CUP B53C22003630006, Italian National Recovery and Resilience Plan funded by NextGenerationEU).

\section*{Author contributions}
All authors conceived ideas and designed the proposed method. F.S.~and L.T.~retrieved and processed data and implemented pipelines. All authors participated to the writing of the manuscript. F.C.~supervised the research.

\section*{Competing interests}
The authors declare no competing interests.

\bibliographystyle{plainnat}
\bibliography{bib}

\clearpage
\setcounter{page}{1}
\appendix
\section*{Supplementary Information}

\renewcommand\thefigure{\thesection.\arabic{figure}}
\renewcommand\thetable{\thesection.\arabic{table}}
\setcounter{figure}{0} 
\setcounter{table}{0}

\section{Technical Lemmas}
\begin{lemma}
\label{lemma:absolute_prob}
    If $| X_n | \pto 0$, then $X_n \pto 0$.
\end{lemma}
\begin{proof}
    We want to show that
    \begin{equation}
        \lim_{n\to\infty} \PP{|X_n - 0| \geq \varepsilon} = 0\,,
    \end{equation}
    where $\varepsilon$ is a positive constant. By definition of convergence in probability of $|X_n|$ to 0, we have that
    \begin{equation}
         \lim_{n\to\infty} \PP{||X_n| - 0| \geq \varepsilon} = 0\,,
    \end{equation}
    where $\varepsilon$ is a positive constant. Now it is easy to notice that
    \begin{equation}
        \PP{||X_n| - 0| \geq \varepsilon}  = \PP{|X_n - 0| \geq \varepsilon}\,,
    \end{equation}
    from which the result follows.
\end{proof}

\section{Proofs of main statements}
\subsection{Proof of Lemma~\ref{th:equivalence}}
\begin{proof}
We start by showing that $\EE{\mathcal{Y}^{(1)}\mid X=x} = \EE{{\gammatarget}^{(1)}(\data) \mid X=x}$. This is done exploiting the identifiability Assumptions~\ref{ass:identify}. Indeed
\begin{equation}
    \begin{split}
        \EE{\gamma^{(1)}(\data)\mid X=x} &=\EE{{\mutarget}^{(1)}(X) + \frac{A}{\pitarget(X)} \left(\mathcal{Y} - {\mutarget}^{(1)}(X) \right) \mid X=x} \\
        &= \EE{{\mutarget}^{(1)}(X) \mid X=x} + \EE{\frac{A}{\pitarget(X)} \left(\mathcal{Y} - {\mutarget}^{(1)}(X) \right) \mid X=x} \\
        &= \EE{\EE{\mathcal{Y}\,|\,X,A=1}\mid X=x} + \EE{\frac{A}{\pitarget(X)} \left(\mathcal{Y} - {\mutarget}^{(1)}(X) \right) \mid X=x} \\
        &= \EE{\EE{\mathcal{Y}\,|\,X,A=1}\mid X=x} \\
        &= \EE{\mathcal{Y}^{(1)}\mid X=x}\,.
    \end{split}
\end{equation}
A similar computation leads to $\EE{\mathcal{Y}^{(0)}\mid X=x} = \EE{{\gammatarget}^{(0)}(\data)\mid X=x}$. Combining the two results implies the desired equality.
\end{proof}

\subsection{Proof of Theorem~\ref{th:oracle}}
\begin{proof}
    First, we show the conditional bias decomposition. By iterated expectations, it holds that
    \begin{equation}
        \begin{split}
            \hat b^{[-j]}(x) & = \left(\frac{\pitarget(x)}{\hat\pi^{[-j]}(x)}- 1 \right) \left({\mutarget}^{(1)} - \hat\mu^{(1)[-j]}(x) \right) - \left(\frac{1 - \pitarget(x)}{1- \hat\pi^{[-j]}(x)}- 1 \right) \left({\mutarget}^{(0)} - \hat\mu^{(0)[-j]}(x) \right) \\
            &= \frac{\left(\hat\pi^{[-j]}(x) - \pitarget(x)\right)\left(\hat\mu^{(1)[-j]}(x) - {\mutarget}^{[1]}(x)\right)}{\hat\pi(x)} + \frac{\left(\hat\pi^{[-j]}(x) - \pitarget(x)\right)\left(\hat\mu^{(0)[-j]}(x) - {\mutarget}^{[0]}(x)\right)}{1- \hat\pi(x)}\,.
        \end{split}
    \end{equation}    
    Now, the assumption of stability implies that 
    \begin{equation}
        \frac{\norm{\hat\theta^{[j]}(x) - \Tilde{\theta}^{[j]}(x) - \empEE{[j]}{\hat b^{[-j]}(X)\mid X=x}}}{\Tilde{R}^{[j]}_n(x)} \pto 0\,.
    \end{equation}
    By reverse triangle inequality, we know that 
    \begin{equation}
        \norm{\hat\theta^{[j]}(x) - \Tilde{\theta}^{[j]}(x) - \empEE{[j]}{\hat b^{[-j]}(X)\mid X=x}} \geq \left| \norm{\hat\theta^{[j]}(x) - \Tilde{\theta}^{[j]}}(x) - \norm{\empEE{[j]}{\hat b^{[-j]}(X)\mid X=x}} \right|\,,
    \end{equation}
    which by Squeeze Theorem in turn implies
    \begin{equation}
        \frac{\left| \norm{\hat\theta^{[j]}(x) - \Tilde{\theta}^{[j]}(x)} - \norm{\empEE{[j]}{\hat b^{[-j]}(X)\mid X=x}} \right|}{\Tilde{R}^{[j]}_n(x)} \pto 0\,.
    \end{equation}
    Combined with Lemma~\ref{lemma:absolute_prob} and consistency, this can be written as
    \begin{equation}
        \norm{\hat\theta^{[j]}(x) - \Tilde{\theta}^{[j]}(x)} = \norm{\empEE{[j]}{\hat b^{[-j]}(X)\mid X=x}} + o_\mathbbmss{P}\left(R^{[j]}_n(x) \right)\,.
    \end{equation}
    If $\norm{\empEE{[j]}{\hat b^{[-j]}\mid X=x}} = o_\mathbbmss{P}\left(\Tilde{R}^{[j]}_n(x)\right)$, the result follows.
\end{proof}

\subsection{Proof of Proposition~\ref{th:inf}}
\begin{proof}
    Stability and $d\left(\hat\gamma^{(a)}_{\hat\mu^{(a)[-j]};\hat\pi^{[-j]}}(\data),{\gammatarget}^{(a)}(\data)\right)\pto 0$ for $a\in\{0,1\}$ together imply that $\hat\Sigma(x) \pto \Sigma(x)$. Therefore, by Slutsky theorem, the result follows.
\end{proof}

\section{Additional details on simulations}

\begin{figure}[h]
    \centering
    \includegraphics[width=0.5\linewidth]{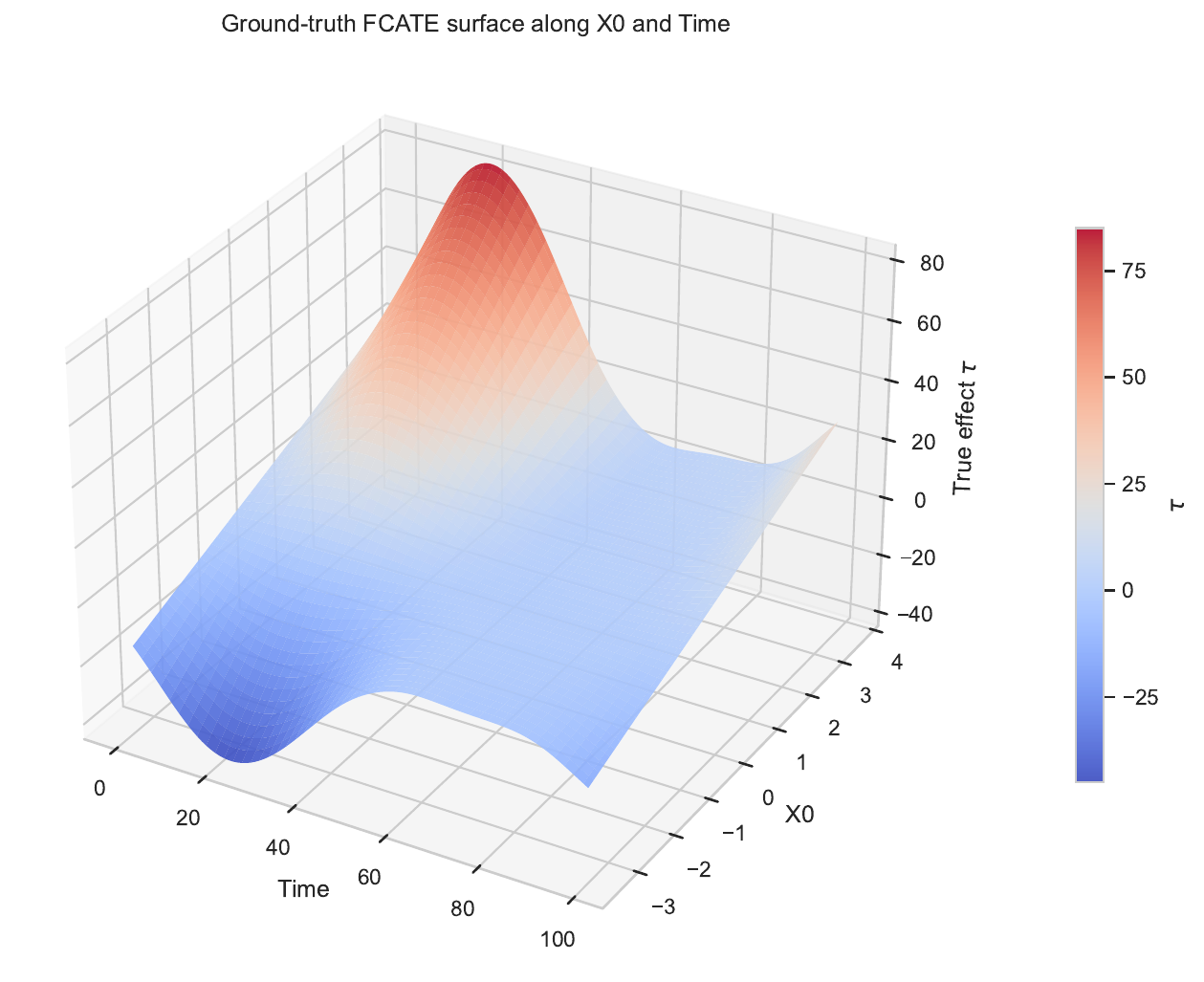}
    \caption{Ground truth F-CATE surface for simulation study. This surface represents the true Functional Conditional Average Treatment Effect (F-CATE) used as the target for one run of the simulation experiments. It is visualized as a function of the functional domain (Time) and the primary covariate of interest ($X_0$), illustrating the complex, non-linear heterogeneity the \FOCAL{} framework is designed to recover.}
    \label{fig:true_surface}
\end{figure}

\clearpage
\section{Additional details on real-data applications}

\subsection{SHARE study}

\begin{table}[h]
    \centering
    \caption{Variables employed in the SHARE application study.}
    \label{tab:covariates}
    \begin{tabular}{lll}
        \toprule
        \textbf{Variable name} & \textbf{Type} & \textbf{Usage} \\
        \midrule
        Hypertension condition & Binary & Treatment \\
        High cholesterol & Binary & Treatment \\
        Age & Scalar & Covariate \\
        Years of education & Scalar & Covariate \\
        Number of children & Scalar & Covariate \\
        Gender & Binary & Covariate \\
        Smoke & Binary & Covariate \\
        Vaccinations during childwood & Binary & Covariate \\
        Mobility index & Functional & Outcome \\
        CASP & Functional & Outcome \\
        \bottomrule
    \end{tabular}
\end{table}

\begin{figure}[h]
    \centering
    \includegraphics[width=0.5\linewidth]{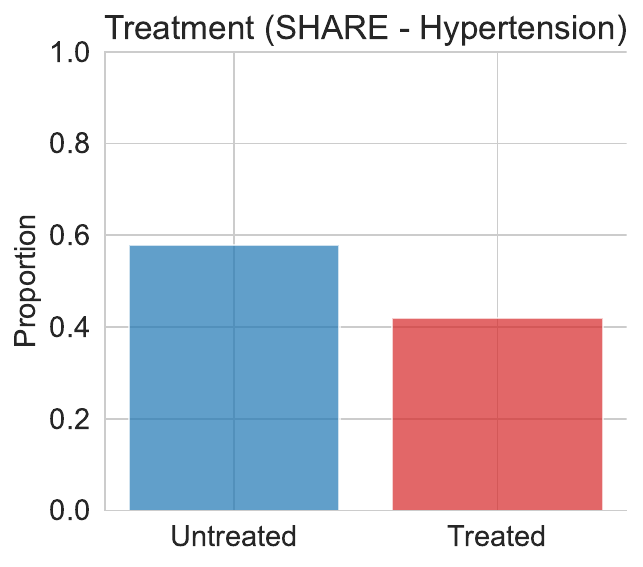}
    \caption{Distribution of treatment status for the SHARE hypertension study. Bar chart showing the proportion of subjects who presented with hypertension at the beginning of the study (Treated, $n=419$) compared to those who remained healthy throughout the observation window (Untreated, $n=577$).}
    \label{fig:thyp}
\end{figure}

\begin{figure}[h]
    \centering
    \includegraphics[width=0.5\linewidth]{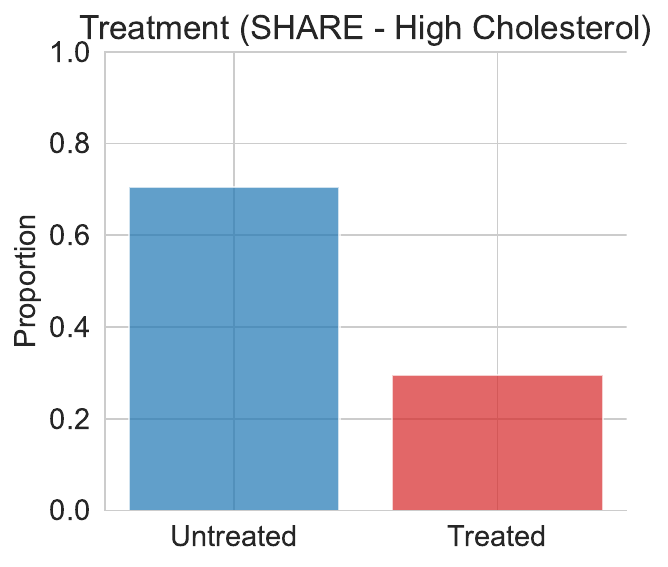}
    \caption{Distribution of treatment status for the SHARE high cholesterol study. Bar chart illustrating the proportion of subjects in the treated (high cholesterol, $n=313$) and untreated ($n=747$) cohorts for the secondary SHARE application.}
    \label{fig:thc}
\end{figure}

\begin{figure}[h]
    \centering
    \includegraphics[width=\linewidth]{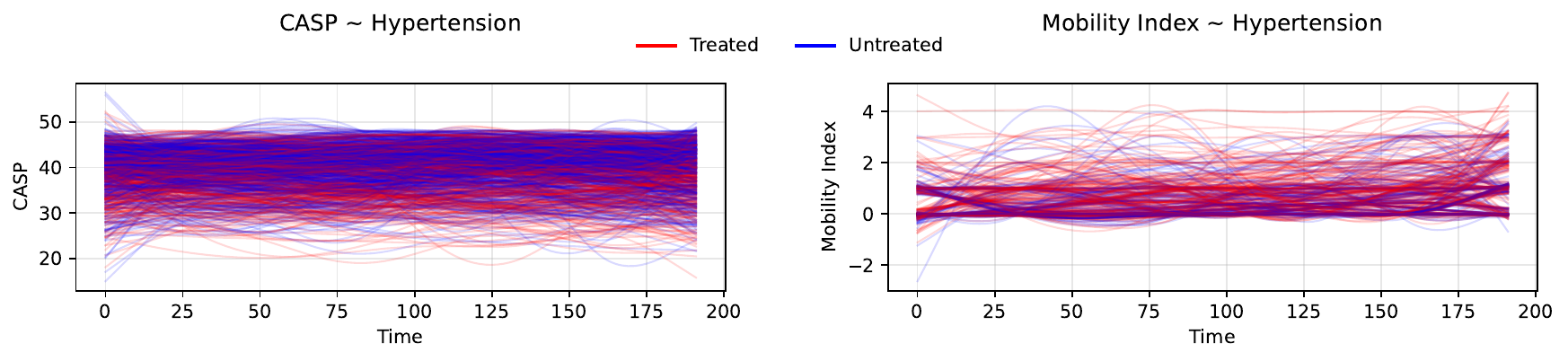}
    \caption{Observed CASP and mobility index functional outcome trajectories for the hypertension cohort (SHARE).}
    \label{fig:outhyp}
\end{figure}

\begin{figure}[h]
    \centering
    \includegraphics[width=\linewidth]{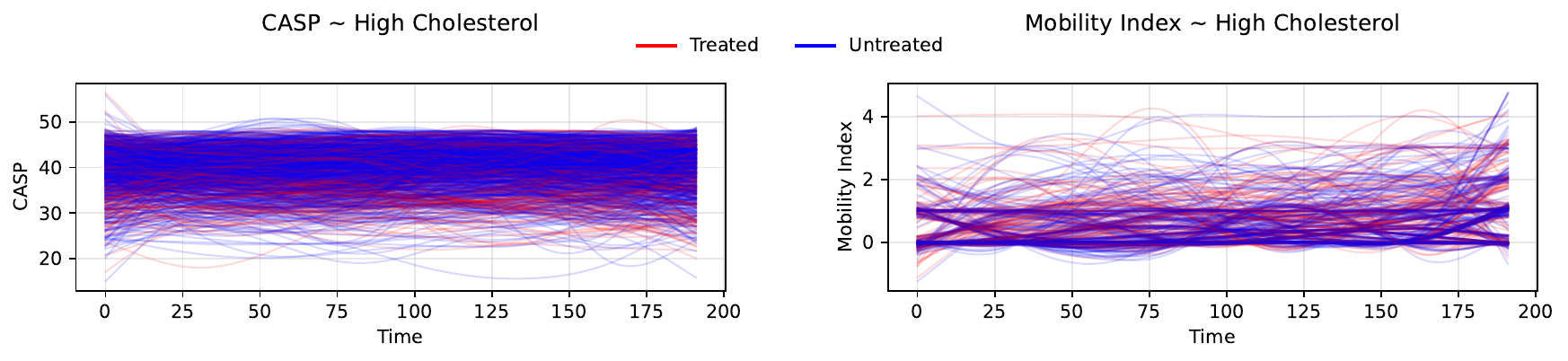}
    \caption{Observed CASP and mobility index functional outcome trajectories for the high cholesterol cohort (SHARE).}
    \label{fig:outchol}
\end{figure}

\begin{figure}[h]
    \centering
    \includegraphics[width=\linewidth]{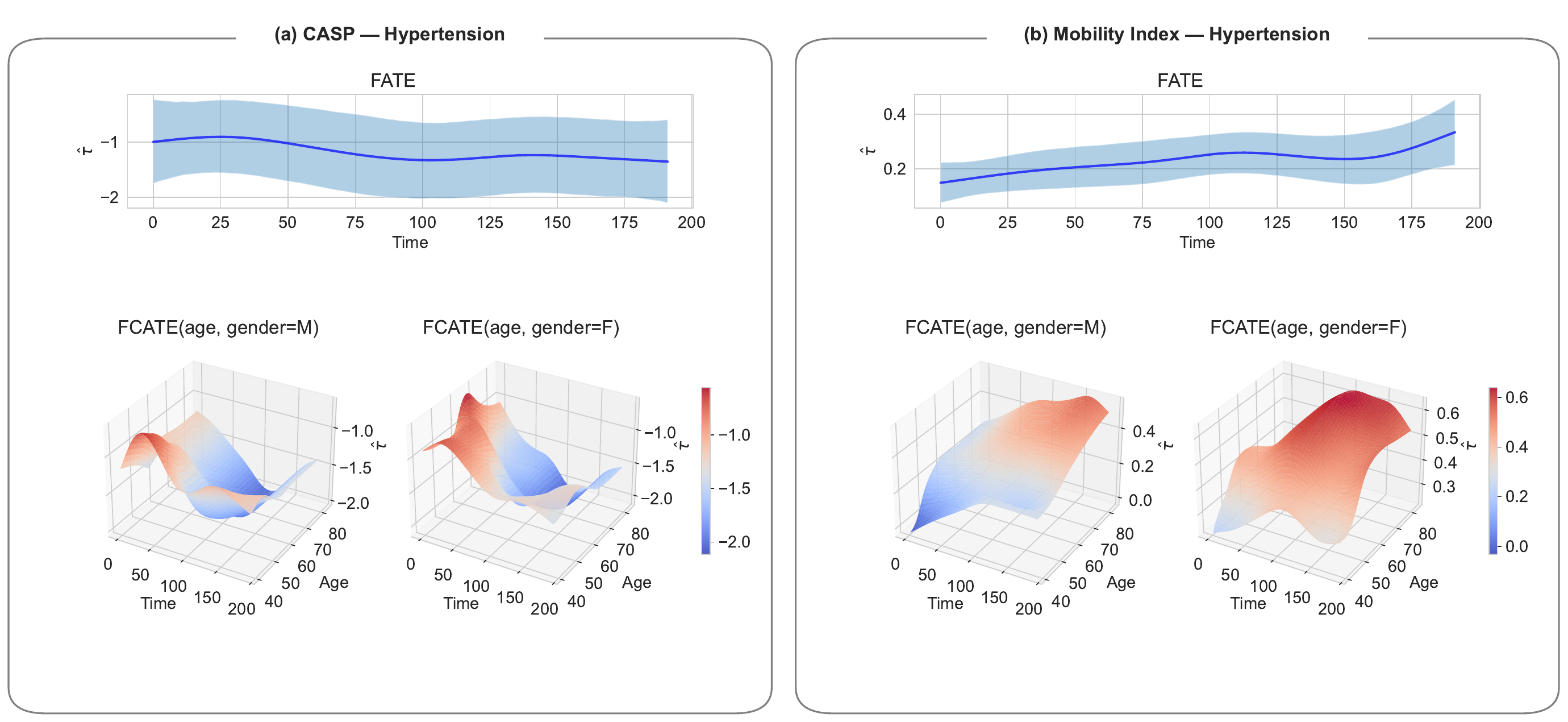}
    \caption{Additional results for the impact of hypertension on quality of life. Compared to Figure \ref{fig:share}, here we set the conditioning variable education to a lower quantile (0.2), keeping all other binary and scalar covariates at their modal and mean values, respectively. }
    \label{fig:edu}
\end{figure}

\begin{figure}[h]
    \centering
    \includegraphics[width=\linewidth]{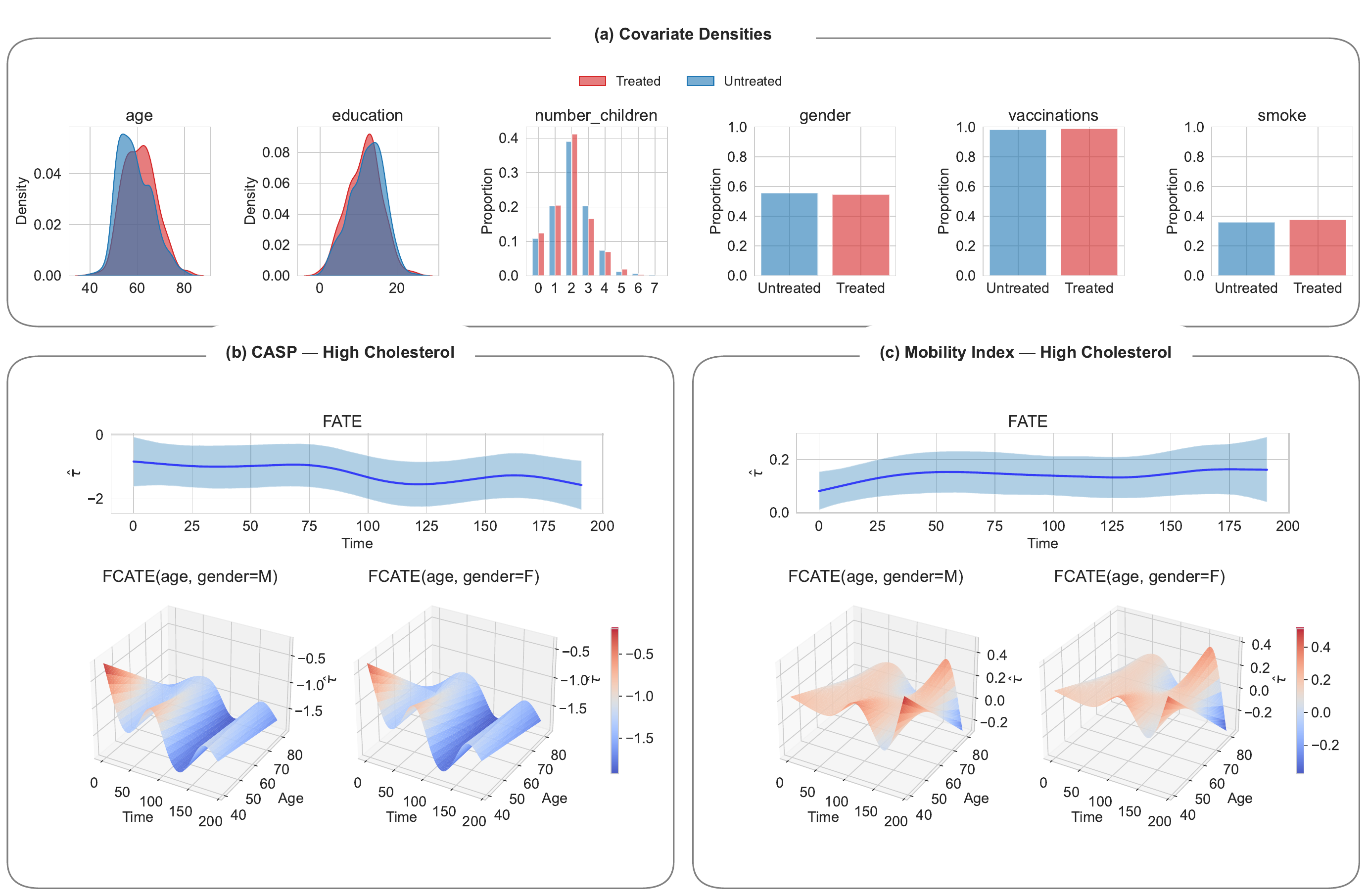}
    \caption{Heterogeneous impact of high cholesterol on quality of life trajectories (SHARE study). Panels can be interpreted as in Figure \ref{fig:share}.
    Note: the plots for CASP and the Mobility Index have response-specific vertical scales (the two outcomes are measured on different scales; see Supplementary Figure \ref{fig:outchol}).}
    \label{fig:chol}
\end{figure}

\clearpage
\subsection{COVID-19 study}

\begin{table}[h]
    \centering
    \caption{Variables employed in the COVID-19 application study.}
    \label{tab:covariates_covid}
    \begin{tabular}{lll}
        \toprule
        \textbf{Variable name} & \textbf{Type} & \textbf{Usage} \\
        \midrule
        Distributed primary health care (number of adults per family doctor) & Binary & Treatment \\
        Area before & Scalar & Covariate \\
        Percentage of population over 65 & Scalar & Covariate \\
        Average beds per hospital & Scalar & Covariate \\
        Average students per classroom & Scalar & Covariate \\
        Average employees per firm & Scalar & Covariate \\
        PM10 & Scalar & Covariate \\
        Mortality rate & Functional & Outcome \\        \bottomrule
    \end{tabular}
\end{table}

\begin{figure}[h]
    \centering
    \includegraphics[width=0.5\linewidth]{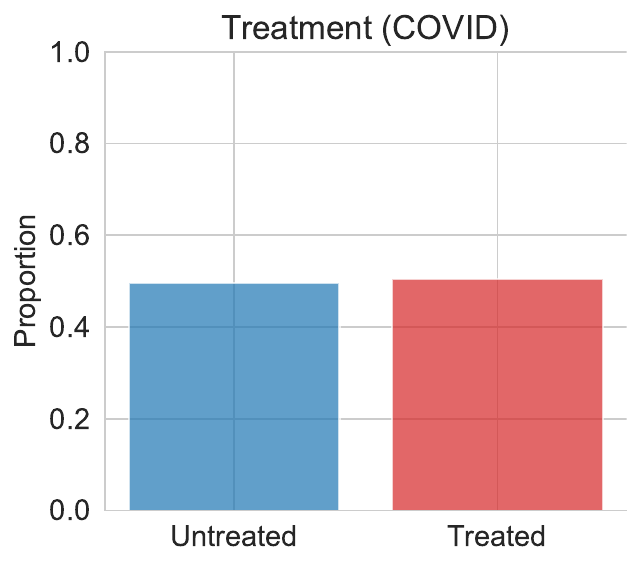}
    \caption{Distribution of treatment status for the COVID-19 study. Proportion of Italian provinces classified as ``Treated'' (poorer distributed primary health care, defined as being above the national median for the ratio of adults to family doctors, $n=54$) and ``Untreated'' (better distributed primary care, $n=53$).}
    \label{fig:tcov}
\end{figure}

\begin{figure}[h]
    \centering
    \includegraphics[width=\linewidth]{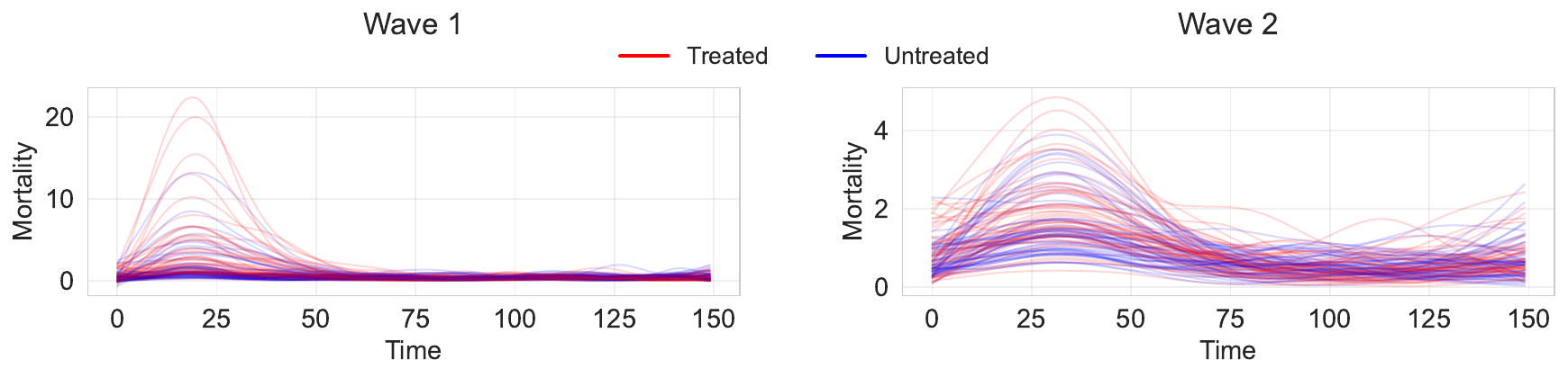}
    \caption{Observed COVID-19 mortality functional outcome trajectories during the first and second waves.}
    \label{fig:outcovid}
\end{figure}

\begin{figure}[h]
    \centering
    \includegraphics[width=\linewidth]{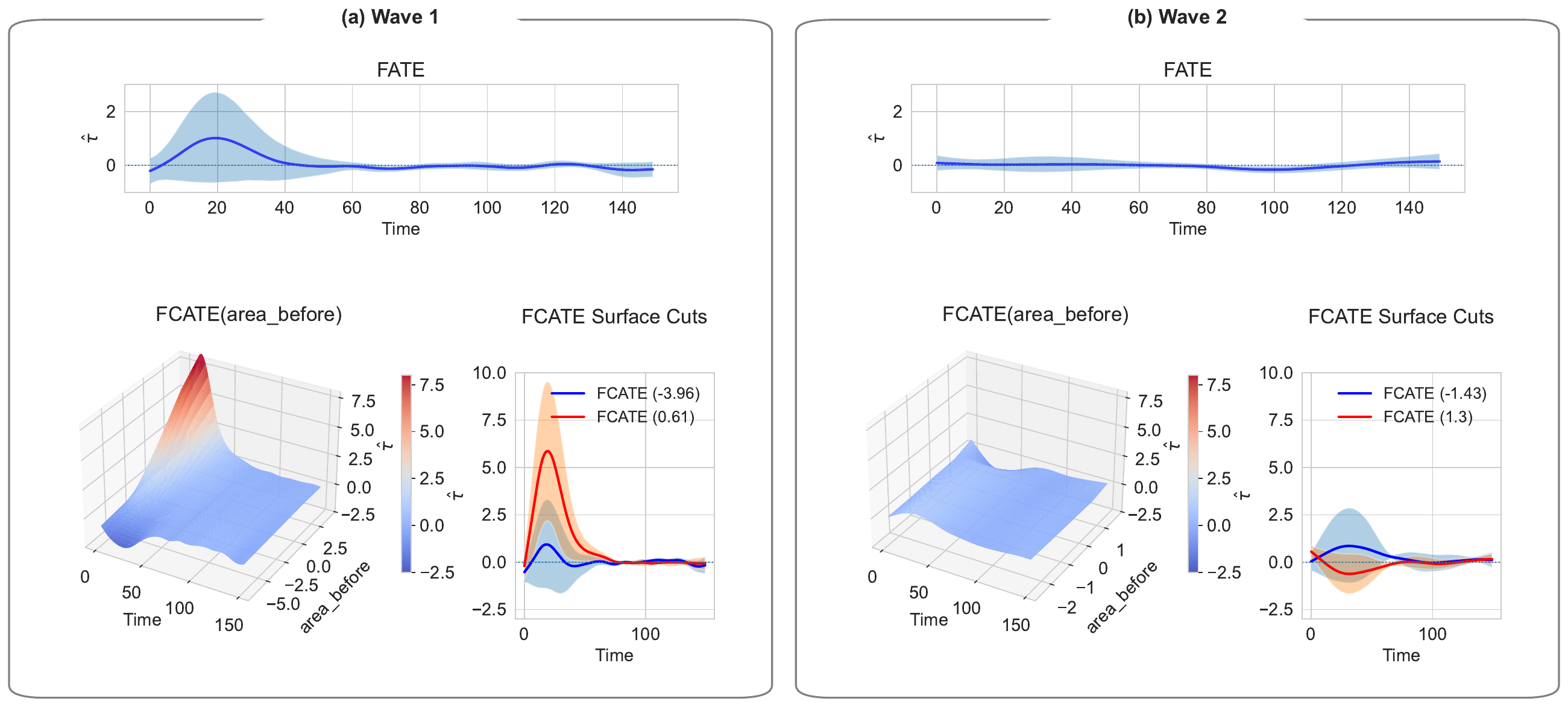}
    \caption{Additional results on the heterogeneous impact of availability of primary health care on COVID-19 mortality patterns (COVID-19 study). Compared to Figure \ref{fig:covid}, here we set the conditioning variables to a higher quantile (0.6) as a robustness check.}
    \label{fig:eighthq}
\end{figure}

\begin{figure}[h]
    \centering
    \includegraphics[width=\linewidth]{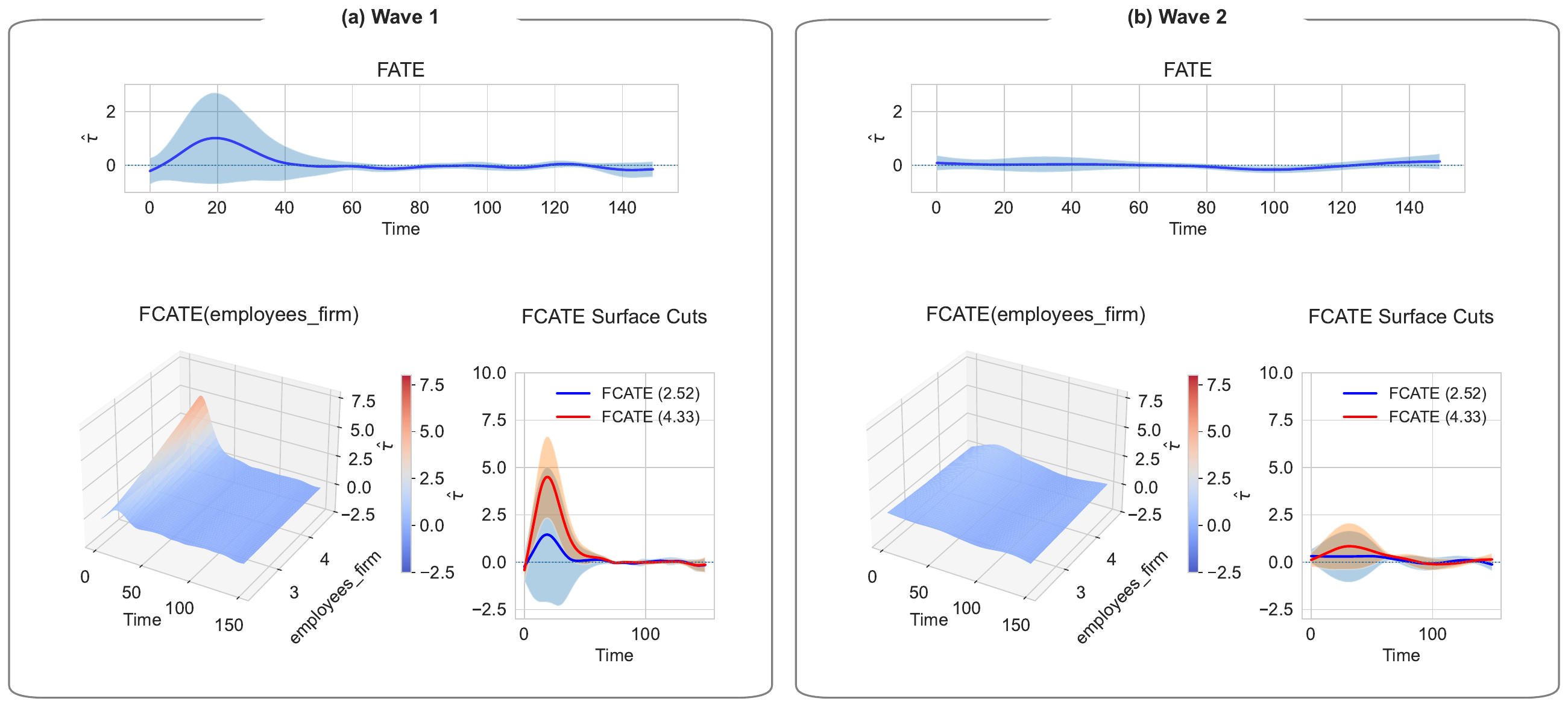}
    \caption{F-CATE analysis of COVID-19 mortality conditioned on workplace. Panels can be interpreted as in Figure \ref{fig:covid}.}
    \label{fig:employees}
\end{figure}

\end{document}